%% file: nmprivacy.tex
\documentclass[11pt]{article}
\usepackage{amsmath,amsthm,bm, amscd}
\usepackage{fullpage}
\usepackage[nofullpage,full,titlepage,nousetoc,nouselot,nouselof,hylinks,final]{boaz}

\usepackage{graphicx}

\newcommand{\ignore}[1]{}

\newcommand{\supp}{\mathsf{supp}}

\providecommand{\lnorm}[2]{\|{#1}\|_{\ell^{#2}}}

\providecommand{\lone}[1]{\lnorm{#1}{1}}
\providecommand{\linfty}[1]{\lnorm{#1}{\infty}}
\providecommand{\ltwo}[1]{\lnorm{#1}{2}}

\newcommand{\zo}{\bits}
\newcommand{\adv}{\calA}
\def\calA{{\mathcal A}} 

\def\calB{{\mathcal B}}

\def\calS{{\mathcal S}}
\def\calT{{\mathcal T}}

\def\dbC{{\mathbb C}}

\def\dbR{{\mathbb R}}
\def\dbZ{{\mathbb Z}}
\DeclareMathOperator{\nm}{\nmExt} 
\DeclareMathOperator{\expect}{E}

\def\chibar{{\overline \chi}} 
\newcommand{\dlog}{\log}
\newcommand{\purify}{\mathsf{purify}}

\theoremstyle{definition}

\newcommand{\set}[1]{\left \{{#1} \right \}}

\newcommand{\eps}{\epsilon}

\newcommand{\Cond}{\mathsf{Cond}}

\newcommand{\Ext}{\mathsf{Ext}}

\newcommand{\nmExt}{\mathsf{nmExt}}
\newcommand{\mac}{\mathsf{MAC}}
\newcommand{\lrmac}{\mathsf{lrMAC}}

\newcommand{\BI}{\begin{itemize}}
\newcommand{\EI}{\end{itemize}}
\newcommand{\BE}{\begin{enumerate}}
\newcommand{\EE}{\end{enumerate}}

\newtheorem{thm}{Theorem}      
\newcommand{\BT}{\begin{theorem}}   \newcommand{\ET}{\end{theorem}}
\newcommand{\BD}{\begin{definition}}   \newcommand{\ED}{\end{definition}}
\newcommand{\BCR}{\begin{corollary}} \newcommand{\ECR}{\end{corollary}}
\newtheorem{constr}[thm]{Construction}
\newcommand{\BCT}{\begin{constr}} \newcommand{\ECT}{\end{constr}}
\newcommand{\BL}{\begin{lemma}}   \newcommand{\EL}{\end{lemma}}

\newcommand{\BP}{\begin{proposition}}   \newcommand{\EP}{\end{proposition}}
\newcommand{\BCM}{\begin{claim}}   \newcommand{\ECM}{\end{claim}}
\newcommand{\BF}{\begin{fact}}   \newcommand{\EF}{\end{fact}}
\newcommand{\BA}{\begin{assumption}}   \newcommand{\EA}{\end{assumption}}

 \def\Tet{{\Theta}}

\def\lam{{\lambda}}

\def\psibar{{\overline \psi}}
 
\def\chibar{{\overline \chi}}

\def\eps{\varepsilon}

\def\le{\leqslant} \def\ge{\geqslant}

\makeatletter
\def\ExtendSymbol#1#2#3#4#5{\ext@arrow 0099{\arrowfill@#1#2#3}{#4}{#5}}
\def\RightExtendSymbol#1#2#3#4#5{\ext@arrow 0359{\arrowfill@#1#2#3}{#4}{#5}}
\def\LeftExtendSymbol#1#2#3#4#5{\ext@arrow 6095{\arrowfill@#1#2#3}{#4}{#5}}
\makeatother
\newcommand\llrightarrow[2][]{\RightExtendSymbol{-}{-}{\rightarrow}{#1}{#2}}

\newcommand\llleftarrow[2][]{\RightExtendSymbol{\leftarrow}{-}{-}{#1}{#2}}

\newcommand{\hinf}{H_\infty}
\newcommand{\thinf}{\widetilde{H}_\infty}

\begin{document}

\begin{titlepage}
\def\thepage{}

\date{}
\title{
Privacy Amplification and Non-Malleable Extractors\\ Via Character Sums
}

\author{
Yevgeniy Dodis\thanks{
Partially supported by
NSF Grants CNS-1065288, CNS-1017471, CNS-0831299 and Google Faculty Award.}\\
Department of Computer Science\\
New York University\\
New York, NY 10012, U.S.A.\\
dodis@cs.nyu.edu\\
\and
Xin Li\thanks{
Partially supported by
NSF Grants CCF-0634811 and CCF-0916160 and THECB ARP Grant 003658-0113-2007.}\\
Department of Computer Science\\
University of Texas at Austin\\
Austin, TX 78701, U.S.A.\\
lixints@cs.utexas.edu\\
\and
Trevor D. Wooley\thanks{Supported
by a Royal Society Wolfson Research Merit Award.}\\
School of Mathematics \\
University of Bristol \\
University Walk, Clifton\\
Bristol BS8 1TW, United Kingdom\\
matdw@bristol.ac.uk\\
\and
David Zuckerman\footnotemark[2]
\\
Department of Computer Science\\
University of Texas at Austin\\
1616 Guadalupe, Suite 2.408\\
Austin, TX 78701, U.S.A.\\
diz@cs.utexas.edu
}

\maketitle \thispagestyle{empty}

\begin{abstract}
In studying how to communicate over a public channel with an active adversary, Dodis and Wichs introduced the notion of a non-malleable extractor. A non-malleable extractor dramatically strengthens the notion of a strong extractor.  A strong extractor takes two inputs, a weakly-random $x$ and a uniformly random seed $y$, and outputs a string which appears uniform, even given $y$.  For a non-malleable extractor $\nm$, the output $\nm(x,y)$ should appear uniform given $y$ as well as $\nm(x,\adv(y))$, where $\adv$ is an arbitrary function with $\adv(y) \neq y$.

We show that an extractor introduced by Chor and Goldreich is non-malleable when the entropy rate is above half.
It outputs a linear number of bits when the entropy rate is $1/2 + \alpha$, for any $\alpha>0$.
Previously, no nontrivial parameters were known for any non-malleable extractor.
To achieve a polynomial running time when outputting many bits, we rely on a widely-believed conjecture about the distribution of prime numbers
in arithmetic progressions.\ Our analysis involves character sum estimates, which may be of independent interest.

Using our non-malleable extractor, we obtain protocols for ``privacy amplification":  key agreement between two parties who share a weakly-random secret. Our protocols work in the presence of an active adversary with unlimited computational power, and have asymptotically optimal entropy loss. When the secret has entropy rate greater than $1/2$, the protocol follows from a result of Dodis and Wichs, and takes two rounds.  When the secret has entropy rate $\delta$ for any constant~$\delta>0$, our new protocol takes a constant (polynomial in $1/\delta$) number of rounds. Our protocols run in polynomial time under the above well-known conjecture about primes.
\end{abstract}
\end{titlepage}

\section{Introduction}
Bennett, Brassard, and Robert \cite{bbr} introduced the basic cryptographic question of \emph{privacy amplification}.
Suppose Alice and Bob share an $n$-bit secret key~$X$, which is weakly random.
This could occur because the secret is a password or biometric data, neither of which is uniformly random, or because an adversary Eve managed
to learn some information about a secret which previously was uniformly random.
How can Alice and Bob communicate over a public
channel to transform $X$ into a nearly uniform secret key, 
about which Eve has negligible information?
We measure the randomness in $X$ using min-entropy.

\begin{definition}
The \emph{min-entropy} of a random variable~$X$ is
\[ \hinf(X)=\min_{x \in \supp(X)}\log_2(1/\Pr[X=x]).\]
For $X \in \zo^n$, we call $X$ an $(n,\hinf(X))$-source, and we say $X$ has
\emph{entropy rate} $\hinf(X)/n$.
\end{definition}

We assume Eve has unlimited computational power.
If Eve is passive, i.e., cannot corrupt the communication
between Alice and Bob, then it is not hard to use randomness extractors~\cite{nz} to solve this problem.
In particular, a strong extractor suffices.

\smallskip\noindent
{\bf Notation.} We let $[s]$ denote the set $\set{1,2,\ldots,s}$.
For $\ell$ a positive integer,
$U_\ell$ denotes the uniform distribution on $\zo^\ell$, and
for $S$ a set,
$U_S$ denotes the uniform distribution on $S$.
When used as a component in a vector, each $U_\ell$ or $U_S$ is assumed independent of the other components.
We say $W \approx_\eps Z$ if the random variables $W$ and $Z$ have distributions which are $\eps$-close in variation distance.

\begin{definition}\label{def:strongext}
A function $\Ext : \bits^n \times \bits^d \rightarrow \bits^m$ is  a \emph{strong $(k,\eps)$-extractor} if for every source $X$ with min-entropy $k$
and independent $Y$ which is uniform on $\zo^d$,
\[ (\Ext(X, Y), Y) \approx_\eps (U_m, Y).\]
\end{definition}

Using such an extractor, the case when Eve is passive can be solved as follows. Alice chooses a fresh random string $Y$ and sends it to Bob. They then both compute $\Ext(X, Y)$. The property of the strong extractor guarantees that even given $Y$, the output is close to uniform.

The case when Eve is active, i.e., can corrupt the communication, has recently received attention. Maurer and Wolf \cite{MW97} gave a one-round protocol which works when the entropy rate of the weakly-random secret $X$ is bigger than $2/3$. This was later improved by Dodis, Katz, Reyzin, and Smith \cite{dkrs} to work for entropy rate bigger than $1/2$. However in both cases the resulting nearly-uniform secret key $R$ is significantly shorter than the min-entropy of $X$.
Dodis and Wichs \cite{DodW} showed that there is no one-round protocol for entropy rate less than $1/2$.
Renner and Wolf \cite{RW03} gave the first protocol which works for entropy rate below $1/2$.
Kanukurthi and Reyzin \cite{kr:agree-close} simplified their protocol and showed that the protocol
can run in $O(s)$ rounds and achieve entropy loss $O(s^2)$ to achieve security parameter~$s$.
(Recall that a protocol achieves security parameter~$s$ if Eve cannot predict with advantage
more than $2^{-s}$ over random.  For an active adversary, we further require that Eve cannot force Alice and Bob to output different secrets and not abort
with probability more than $2^{-s}$.)
Dodis and Wichs~\cite{DodW} improved the number of rounds to 2 but did not improve the entropy loss.
Chandran, Kanukurthi, Ostrovsky, and Reyzin \cite{ckor} improved the entropy loss to $O(s)$ but the number of rounds remained $O(s)$.
The natural open question is therefore whether there is a 2-round protocol with entropy loss $O(s)$.

Dodis and Wichs showed how such a protocol could be built using
\emph{non-malleable extractors}, which they defined.
In the following definition of (worst-case) non-malleable extractor, think of an adversary changing the value of the seed via the function $\adv$.

\begin{definition}
\label{nmdef}
A function $\nm:[N] \times [D] \to [M]$ is a $(k,\eps)$-non-malleable extractor if,
for any source $X$ with $\hinf(X) \geq k$ and any function $\adv:[D] \to [D]$ such that $\adv(y) \neq y$ for all~$y$,
the following holds.
When $Y$ is chosen uniformly from $[D]$ and independent of $X$,
\[
(\nm(X,Y),\nm(X,\adv(Y)),Y) \approx_\eps (U_{[M]},\nm(X,\adv(Y)),Y).
\]
\end{definition}

Note that this dramatically strengthens the definition of strong extractor.
In a strong extractor, the output must be indistinguishable from uniform, even given the random seed.
For a non-malleable extractor, a distinguisher is not only given a random seed, but also the output of the extractor with the given
input and an arbitrarily correlated random seed.  Note that $\nm(X,\adv(Y))$ need not be close to uniform.
The above ``worst-case'' definition is slightly weaker than the ``average-case'' definition needed by applications, but Dodis and Wichs showed that any worst-case $(k,\eps)$-non-malleable extractor is also an average-case $(k-\log(1/\eps),2\eps)$-non-malleable extractor. See Subsection~\ref{avgcase}.

Unfortunately, Dodis and Wichs were not able to construct such non-malleable extractors.
Instead, they constructed ``look-ahead extractors," which are weaker than non-malleable extractors, but nevertheless yielded
the two-round, $O(s^2)$-entropy loss protocol mentioned above.

Dodis and Wichs also showed the existence of non-malleable extractors.
The existence of excellent standard randomness extractors can be shown by the probabilistic method in a straightforward way.
For non-malleable extractors, the argument requires more work. Nevertheless, Dodis and Wichs showed that non-malleable extractors exist
with $k>2m+3\log(1/\eps) + \log d + 9$ and $d>\log(n-k+1) + 2\log (1/\eps) + 7$, for $N=2^n$, $M=2^m$, and $D=2^d$.

The definition of non-malleable extractor is so strong that before our work, no explicit construction was known for any length seed achieving
a one-bit output, even for min-entropy $k=.99n$.
For example, a first attempt might be $f(x,y) = x \cdot y$, where the inner product is taken over GF(2).  However, this fails, even for min-entropy $n-1$.  To see this, take $X$ to be the bit 0 concatenated with $U_{n-1}$.  Let $\adv(y)$ be $y$ with the first bit flipped.  Then for all $x$ in the support of $X$, one has $f(x,y) = f(x,\adv(y))$.

Although general Hadamard codes don't work, we nevertheless show that a specific near-Hadamard code that comes from the Paley graph
works for min-entropy $k>n/2$.  The Paley graph function is $\nm(x,y) = \chi(x-y)$, where $x$ and $y$ are viewed as elements in a finite field $\F$ of odd order~$q$ and $\chi$ is the
quadratic character $\chi(x) = x^{(q-1)/2}$.
(The output of $\chi$ is in $\set{\pm 1}$, which we convert to an element of $\zo$.)
The function $\nm(x,y) = \chi(x+y)$ works equally well.
The proof involves estimating a nontrivial character sum.

We can output $m$ bits by computing the discrete logarithm $\log_g (x+y) \mod M$.  
This extractor was originally introduced by Chor and Goldreich \cite{cg:weak} in the context of two-source extractors.
To make this efficient, we need $M$ to divide $q-1$.
A widely-believed conjecture about primes in arithmetic progressions implies that such a $q$ is not too large (see Conjecture~\ref{conj-primes}).
Our result is stated as follows.

\BT \label{thm:nm}
For any constants $\alpha,\beta,\gamma>0$ with $\beta + \gamma < \alpha/2$,
there is an explicit $(k=(1/2 + \alpha)n,\eps)$-non-malleable extractor
$\nm:\zo^n \times \zo^d \to \zo^m$ for $\eps = 2^{-\gamma n}$ and any $m \leq \beta n$.
It has seed length $d=n$ and runs in polynomial time if
Conjecture~\ref{conj-primes} holds or $m=O(\log n)$.
\ET

As a direct corollary of Theorem~\ref{thm:nm} and the protocol of Dodis and Wichs, we obtain a 2-round protocol for privacy amplification with optimal entropy loss,
when the entropy rate is $1/2+\alpha$ for any $\alpha > 0$.  This improves the significant entropy loss in the one-round protocols
of Dodis, Katz, Reyzin, and Smith \cite{dkrs} and Kanukurthi and Reyzin \cite{kr:robust-fuzzy}.

Next, we use our non-malleable extractor to give a constant-round privacy amplification protocol with optimal entropy loss, when the entropy rate is $\delta$ for any constant $\delta>0$.
This significantly improves the round complexity of  \cite{kr:agree-close} and \cite{ckor}. It also significantly improves the entropy loss of \cite{DodW}, at the price of a larger, but still comparable ($O(1)$ vs. 2), round complexity.
Our result is stated as follows.

\BT \label{thm:privacy}
Under conjecture~\ref{conj-primes}, for any constant $0<\delta<1$ and error $2^{-\Omega(\delta n)} < \e < 1/n$, there exists a
polynomial-time, constant-round $(k=\delta n, m=\delta n - O(\log(1/\e)), \eps)$-secure
protocol for privacy amplification. More specifically, the protocol takes number of rounds $\poly(1/\delta)=O(1)$, and achieves entropy loss $k-m = \poly(1/\delta)\log(1/\e)=O(\log(1/\e))$.
\ET

\noindent{\bf Subsequent work.}
Following the preliminary version of our work \cite{dlwz-focs},
Cohen, Raz, and Segev \cite{CRS11} gave an alternative construction of a non-malleable extractor
for min-entropy rate $1/2 + \alpha$.  Their construction has the advantage that it works for any seed length $d$ with
$2.01 \log n \leq d \leq n$, although their output length $m$ remains small if $d$ is small, i.e., $m = \Theta(d)$.
They further do not rely on any unproven assumption.
Our construction, or at least the one-bit version, appears to be a special case of their construction.

Inspired by their elegant work, we subsequently used ideas related to \cite{CRS11} and \cite{Raz05}
to strengthen our character sum and show that our non-malleable extractor works even if the seed has entropy only $\Theta(m + \log n)$.
In particular, this implies that our extractor can also use a seed as small as $O(\log n)$.
We believe that their proof can also be modified to show that their construction works for weak seeds.

To state our results, we define non-malleable extractors for weak seeds.

\begin{definition}\label{def:weak-seed}
A function $\nm:[N] \times [D] \to [M]$ is a $(k,k',\eps)$-non-malleable extractor if,
for any source $X$ with $\hinf(X) \geq k$, any seed $Y$ with $\hinf(Y) \geq k'$,
and any function $\adv:[D] \to [D]$ such that $\adv(y) \neq y$ for all $y$,
the following holds:
\[
(\nm(X,Y),\nm(X,\adv(Y)),Y) \approx_\eps (U_{[M]},\nm(X,\adv(Y)),Y).
\]
\end{definition}

We can now state our theorem for weak seeds.  We stress that we proved this theorem only after seeing \cite{CRS11}.

\begin{theorem}
\label{nm-weakseed}
For any $\epsilon > 0$ and constant $\alpha>0$, there is a constant $c \leq 8/\alpha$ such that
there is an explicit $(k=(1/2+\alpha)n, k' , \eps)$-non-malleable extractor
$\nm:\zo^n \times \zo^d \to \zo^m$
for $d=n$ and $k' = c(m+ \log \eps^{-1} + \log n)$. In particular, we can reduce the seed length $d$ of our $(k=(1/2+\alpha)n, \eps)$-non-malleable extractor to $d=c(m+ \log \eps^{-1} + \log n)$.
Our extractor runs in polynomial time if Conjecture~\ref{conj-primes} holds or $m=O(\log n)$.
\end{theorem}

\medskip
\noindent{\bf Organization.}
Since our first proof is for the non-malleable extractor, we begin with an overview of our privacy amplification protocol in Section~\ref{overview}.  (Readers interested only in the non-malleable extractor can skip this section.)  We discuss some preliminaries in Section~\ref{prelim}, the non-malleable extractor in Section~\ref{ext}, and the character sum estimate in Section~\ref{csum}.  Finally, we give full details of the privacy amplification protocol in Section~\ref{privacy}. In Appendix~\ref{gen} we give a generalization of the non-malleable extractor.

\section{Overview of the Protocol for Privacy Amplification}\label{overview}

We first describe Dodis and Wichs' optimal two-round protocol using a non-malleable extractor. The protocol also uses a cryptographic primitive: a one-time message authentication code (MAC). Roughly speaking, a MAC uses a private uniformly random key $R$ to produce a tag $T$ for a message $m$, such that without knowing the key, the probability that an adversary can guess the correct tag $T'$ for another message $m' \neq m$ is small, even given $m$ and $T$.

Now assume that we have a non-malleable extractor $\nmExt$ that works for any $(n,k)$-source~$X$. Then there is a very natural two-round privacy amplification protocol. In the first round Alice chooses a fresh random string $Y$ and sends it to Bob. Bob receives a possibly modified string $Y'$. They then compute $R=\nmExt(X, Y)$ and $R'=\nmExt(X, Y')$ respectively. In the second round, Bob chooses a fresh random string $W'$ and sends it to Alice, together with $T'=\mac_{R'}(W')$ by using $R'$ as the MAC key. Alice receives a possibly modified version $(W, T)$, and she checks if $T=\mac_R(W)$. If not, then Alice aborts; otherwise they compute outputs $Z=\Ext(X, W)$ and $Z'=\Ext(X, W')$ respectively, where $\Ext$ is a seeded strong extractor. The protocol is depicted  in Figure~\ref{fig:AKA1}.

The analysis of the above protocol is also simple. If Eve does not change $Y$, then $R=R'$ and is (close to) uniform. Therefore by the property of the MAC the probability that Eve can change $W'$ without being detected is very small. On the other hand if Eve changes $Y$, then by the property of the non-malleable extractor, one finds that $R'$ is (close to) independent of $R$. Thus in this case, again the probability that Eve can change $W'$ without being detected is very small. In fact, in this case Eve cannot even guess the correct MAC for $W'$ with a significant probability.

The above protocol is nice, except that we only have non-malleable extractors for entropy rate $>1/2$. As a direct corollary this gives our 2-round privacy amplification protocol for entropy rate $>1/2$. To get a protocol for arbitrary positive entropy rate, we have to do more work.

We start by converting the shared weak random source $X$ into a somewhere high min-entropy rate source. The conversion uses recent condensers built from sum-product theorems. Specifically, any $n$-bit weak random source with linear min-entropy can be converted into a matrix with a constant number of rows, such that at least one row has entropy rate $0.9$.\footnote{In fact, the result is (close to) a convex combination of such matrices. For simplicity, however, we can assume that the result is just one such matrix, since it does not affect the analysis.} Moreover each row still has $\Theta(n)$ bits. Note that since Alice and Bob apply the same function to the shared weak random source, they now also share the same rows.

Now it is natural to try the two-round protocol for each row and hope that it works on the row with high min-entropy rate. More specifically, for each row $i$ we have a two round protocol that produces $R_i, R'_i$ in the first round and $Z_i, Z'_i$ in the second round. Now let $g$ be the first row that has min-entropy rate $0.9$. We hope that $Z_g=Z'_g$ with high probability, and further that $Z_g,Z'_g$ are close to uniform and private. This is indeed the case if we run the two round protocol for each row sequentially (namely we run it for the first row, and then the second row, the third row, and so on), and can be argued as follows.

Assume the security parameter we need to achieve is $s$, so each of $R_i, R'_i$ has $O(\log n+s)$ bits by the property of the MAC. As long as $s$ is not too large, we can fix all these random variables up to row $g-1$, and argue that row $g$ still has min-entropy rate $>1/2$ (since each row has $\Theta(n)$ bits we can actually achieve a security parameter up to $s=\Omega(n)$). Note that we have essentially fixed all the information about $X$ that can be leaked to Eve. Therefore now for row $g$ the protocol succeeds and thus $Z_g=Z'_g$ with high probability, and $Z_g,Z'_g$  are close to uniform and private.

However, we don't know which row is the good row. We now modify the above protocol to ensure that, once we reach the first good row $g$, for all subsequent rows $i$, with $i >g$, we will have that $Z_i=Z'_i$ with high probability, and further $Z_i, Z'_i$  are close to uniform and private. If this is true then we can just use the output for the last row as the final output.

To achieve this, the crucial observation is that once we reach a row $i-1$ such that $Z_{i-1}=Z'_{i-1}$, and $Z_{i-1},Z'_{i-1}$ are close to uniform and private, then $Z_{i-1}'$ can be used as a MAC key to authenticate $W'_i$ for the next row. Now if $W'_i=W_i$ for row $i$, then $Z_i=Z'_i$ and $Z_i,Z'_i$ will also be close to uniform and private. Therefore, we modify the two-round protocol so that in the second round for row $i$, not only do we use $T'_{i}=\mac_{R'_i}(W'_i)$ to authenticate $W'_i$, but also we use $L'_{i}=\mac_{Z'_{i-1}}(W'_i)$ to authenticate $W'_i$.

This would have worked given that  $Z_{i-1}=Z'_{i-1}$, and $Z_{i-1},Z'_{i-1}$ are close to uniform and private, except for another complication. The problem is that now $T'_{i}=\mac_{R'_i}(W'_i)$ could leak information about $Z_{i-1}$ to Eve, so $Z_{i-1}$ is no longer private. Fortunately, there are known constructions of MACs that work even when the key is not uniform, but instead only has large enough average conditional min-entropy in the adversary's view. Specifically, \theoremref{thm:mac} indicates that the security parameter of this MAC is roughly the average conditional min-entropy of the key minus half the key length, and the key length is roughly twice as long as the length of the tag. Therefore, we can choose a small tag length for  $T'_{i}=\mac_{R'_i}(W'_i)$, and a large tag length for  $L'_{i}=\mac_{Z'_{i-1}}(W'_i)$. For example, if the tag length for $T'_{i}$ is $2s$, and the tag length for $T'_{i2}$ is $4s$, then the key length for $L'_{i}$ is $8s$. Thus the average min-entropy of $Z_{i-1}$ conditioned on $T'_{i}$ is $8s-2s=6s$, and we can still achieve a security parameter of $6s-4s=2s$.

Finally, the discussion so far implicitly assumed that Eve follows a natural ``synchronous'' scheduling, where she never tries to get one party out-of-sync with another party. To solve this  problem, after each Phase $i$ Bob performs a ``liveness'' test, where Alice has to respond to a fresh extractor challenge from Bob to convince Bob that Alice is still ``present'' in this round. This ensures that if Bob completes the protocol, Alice was ``in-sync'' with Bob throughout. However, Eve might be able to make Alice be out-of-sync with Bob, causing Alice to output a non-random key (and Bob reject). To solve this last problem, we add one more round at the end which ensures that Alice always outputs a random key (and Bob either outputs the same key or rejects).

With this modification, the complete protocol is depicted in Figure~\ref{fig:AKA2}. Essentially, for the first good row, the property of the non-malleable extractor guarantees that Eve cannot change $W'_g$ with significant probability. For all subsequent rows, by using the output $Z'_{i-1}$ from the previous row as the MAC key, the property of the MAC guarantees that Eve cannot change $W'_i$ with significant probability. Therefore, the output for the last row can be used to authenticate the last seed of the extractor chosen by Alice (for the reason mentioned above) to produce the final output.

Finally, we note that our final protocol has $O(1)$ rounds and achieves asymptotically optimal
entropy loss is $O(s+\log n)$, for security parameter $s$.

\input{prelim.tex}

\input{extractor.tex}

\input{privacy.tex}

\section{Future Directions}

There are several natural open questions.
First, can we give a non-malleable extractor which outputs even one bit for entropy rate below $1/2$?
As far as we know, it is possible that our extractor works for lower min-entropy
(although the Cohen-Raz-Segev extractor \cite{CRS11} in full generality does not).
Second, can we achieve optimal round complexity (2 rounds) and entropy loss ($O(s)$) for weak secrets with arbitrarily linear entropy $\delta n$? In principle, this problem would be solved if an efficient non-malleable extractor is constructed for entropy rate below $1/2$.
Finally, can we generalize our techniques to sublinear entropy?

\section*{Acknowledgments}

We are grateful to Gil Segev for finding an error in an earlier version of this paper, and to Salil Vadhan
for a helpful discussion. We would also like to thank the anonymous referees for useful comments.

\appendix

\bibliographystyle{alpha}

\bibliography{refs}

\input{appendix.tex}

\end{document}

%% file: prelim.tex
\section{Preliminaries} \label{prelim}
We often use capital letters for random variables and corresponding small letters for their instantiations. Let $|S|$ denote the cardinality of the set~$S$.
Let $\dbZ_r$ denote the cyclic group $\dbZ/(r\dbZ)$,
and let $\F_q$ denote the finite field of size $q$.
All logarithms are to the base 2.

\subsection{Probability distributions}
\begin{definition} [statistical distance]Let $W$ and $Z$ be two distributions on
a set $S$. Their \emph{statistical distance} (variation distance) is
\begin{align*}
\Delta(W,Z) \eqdef \max_{T \subseteq S}(|W(T) - Z(T)|) = \frac{1}{2}
\sum_{s \in S}|W(s)-Z(s)|.
\end{align*}
\end{definition}

We say $W$ is $\eps$-close to $Z$, denoted $W \approx_\eps Z$, if $\Delta(W,Z) \leq \eps$.
For a distribution $D$ on a set $S$ and a function $h:S \to T$, let $h(D)$ denote the distribution on $T$ induced by choosing $x$ according to $D$ and outputting $h(x)$.
We often view a distribution as a function whose value at a sample point is the probability of that sample point.
Thus $\lone{W-Z}$ denotes the $\ell_1$ norm of the difference of the distributions specified by the random variables $W$ and $Z$, which equals $2\Delta(W,Z)$.
\subsection{Average conditional min-entropy}
\label{avgcase}

Dodis and Wichs originally defined non-malleable extractors with respect to average conditional min-entropy, a notion defined by
Dodis, Ostrovsky, Reyzin, and Smith \cite{dors}.

\begin{definition}
The \emph{average conditional min-entropy} is defined as
\[ \thinf(X|W)= - \log \left (\expect_{w \leftarrow W} \left [ \max_x \Pr[X=x|W=w] \right ] \right )
= - \log \left (\expect_{w \leftarrow W} \left [2^{-\hinf(X|W=w)} \right ] \right ).
\]
\end{definition}

Average conditional min-entropy tends to be useful for cryptographic applications.
By taking $W$ to be the empty string, we see that average conditional min-entropy is at least as strong as min-entropy.
In fact, the two are essentially equivalent, up to a small loss in parameters.
We have the following lemmas.

\begin{lemma} [\cite{dors}]
\label{entropies}
For any $s > 0$,
$\Pr_{w \leftarrow W} [\hinf(X|W=w) \geq \thinf(X|W) - s] \geq 1-2^{-s}$.
\end{lemma}

\BL [\cite{dors}] \label{lem:amentropy}
If a random variable $B$ has at most $2^{\ell}$ possible values, then $\thinf(A|B) \geq \hinf(A)-\ell$.
\EL

To clarify which notion of min-entropy and non-malleable extractor we mean, we use the term \emph{worst-case non-malleable extractor} when we refer to
our Definition~\ref{nmdef}, which is with respect to traditional (worst-case) min-entropy, and \emph{average-case non-malleable extractor} to refer to
the original definition of Dodis and Wichs, which is with respect to average conditional min-entropy.

\begin{corollary}
A $(k,\eps)$-average-case non-malleable extractor is a $(k,\eps)$-worst-case non-malleable extractor.
For any $s>0$, a $(k,\eps)$-worst-case non-malleable extractor is a $(k+s,\eps + 2^{-s})$-average-case non-malleable extractor.
\end{corollary}

Throughout the rest of our paper, when we say non-malleable extractor, we refer to the worst-case non-malleable extractor of Definition~\ref{nmdef}.

\subsection{Primes in arithmetic progressions}

To output more than $\log n$ bits, we will rely on a well-known conjecture about primes in arithmetic progressions.  We begin with a definition.

\begin{definition}
Let $p(r,a)$ denote the least prime in the arithmetic progression $a$ modulo $r$.
\end{definition}

We can now state a special case of a well-known conjecture.

\begin{conjecture}
\label{conj-primes}
There exists a constant $c>0$, such that for $r$ a power of 2 and $a=1$, one has $p(r,a) = O(r\log^c r)$.
\end{conjecture}

We don't really need $r$ to be a power of 2; it would suffice if the conjecture held for integers~$r_n$,
where $r_n$ is a smooth integer of about $n$ bits computable in time polynomial in~$n$.
This conjecture is widely believed for $c=2$, all $r$, and all $a$ relatively prime to~$r$.
For more on this conjecture,
see, for example, the discussion following equation (1) of \cite{hb78}.
The best unconditional conclusion is
substantially weaker. Thus, one has $p(r,a) =O(r^{5.2})$ (see \cite{Xyl,hb92}.)

\subsection{Fourier analysis}

The following definitions from Fourier analysis are standard (see e.g., \cite{Ter}) , although we normalize differently than in many computer science papers, such as \cite{Rao:bourgain}.
For functions $f,g$ from a set $S$ to $\dbC$, we define the inner product $\angles{f,g} = \sum_{x \in S} f(x) \overline{g(x)}$.
Let $D$ be a distribution on $S$, which we also view as a function from $S$ to $\dbR$.
Note that $\expect_D [f(D)] = \angles{f,D}$.
Now suppose we have functions $h:S \to T$ and $g:T \to \dbC$.
Then
\[ \angles{g \circ h, D} = \expect_D [g(h(D))] = \angles{g,h(D)}.\]

Let $G$ be a finite abelian group, and let $\phi$ a character of $G$, i.e., a homomorphism from $G$ to $\dbC^\times$.
We call the character that maps all elements to 1 the trivial character.
Define the Fourier coefficient $\widehat{f}(\phi) = \angles{f,\phi}$.  We let $\widehat{f}$ denote the vector with entries $\widehat{f}(\phi)$ for all $\phi$.
Note that for a distribution $D$, one has $\widehat{D}(\phi) = \expect_D[\phi(D)]$.

Since the characters divided by $\sqrt{|G|}$ form an orthonormal basis,
the inner product is preserved up to scale: $\angles{\widehat{f},\widehat{g}} = |G| \angles{f,g}$.  As a corollary, we obtain
Parseval's equality:
\[ \ltwo{\widehat{f}}^2=\angles{\widehat{f},\widehat{f}} = |G| \angles{f,f}=|G|\ltwo{f}^2.\]
Hence by Cauchy-Schwarz,
\begin{equation}
\label{fourier-bound}
\lone{f} \leq \sqrt{|G|} \ltwo{f} = \ltwo{\widehat{f}} \leq \sqrt{|G|} \linfty{\widehat{f}}.
\end{equation}

For functions $f,g:S \to \dbC$, we define the function $(f,g):S \times S \to \dbC$ by $(f,g)(x,y) = f(x)g(y)$.
Thus, the characters of the group $G \times G$ are the functions $(\phi,\phi')$, where $\phi$ and $\phi'$ range over all characters of $G$.
We abbreviate the Fourier coefficient $\widehat{(f,g)}((\phi,\phi'))$ by $\widehat{(f,g)}(\phi,\phi')$.  Note that
\[ \widehat{(f,g)}(\phi,\phi') = \sum_{(x,y) \in G \times G} f(x)g(y)\phi(x)\phi'(y) = \left ( \sum_{x \in G} f(x)\phi(x) \right ) \left ( \sum_{y \in G} g(x)\phi'(x) \right )
= \widehat{f}(\phi) \widehat{g}{(\phi')}. \]

\subsection{A non-uniform XOR lemma}

We'll need the following extension of Vazirani's XOR lemma.  We can't use traditional versions of the XOR lemma, because our output
may not be uniform.  Our statement and proof parallels Rao \cite{Rao:bourgain}.

\begin{lemma}
\label{special-case}
Let $(W,W')$ be a random variable on $G \times G$ for a finite abelian group $G$, and suppose
that for all characters $\phi,\phi'$ on $G$ with $\phi$ nontrivial, one has
$$|\expect_{(W,W')}[\phi(W)\phi'(W')]| \leq \alpha.$$
Then the distribution of $(W,W')$ is $\alpha |G|$ close to $(U,W')$, where $U$ is the uniform distribution on $G$ which is independent of $W'$.
Moreover, for $f:G\times G \to \dbR$ defined as the difference of distributions $(W,W') - (U,W')$,
we have $\linfty{f} \leq \alpha$.
\end{lemma}

\begin{proof}
As implied in the lemma statement, the value of $f(a,b)$ is the probability assigned to $(a,b)$ by the distribution of $(W,W')$ minus
that assigned by $(U,W')$.
First observe that
\[
\widehat{f}(\phi,\phi') = \angles{f,(\phi,\phi')} = \expect_{(W,W')}[\phi(W)\phi'(W')] - \expect_{(U,W')}[\phi(U)\phi'(W')].
\]
Since $U$ and $W'$ are independent, this last term equals
\[\expect_{(U,W')}[\phi(U)]\expect_{(U,W')}[\phi'(W')]
= \expect_{U}[\phi(U)]\expect_{W'}[\phi'(W')] = 0,\]
since $\phi$ is nontrivial.
Therefore, by hypothesis, when $\phi$ is nontrivial, one finds that $|\widehat{f}(\phi,\phi')| \leq \alpha$.

When $\phi$ is trivial, we get
\[ \widehat{f}(\phi,\phi') = \expect_{(W,W')}[\phi'(W')] - \expect_{(U,W')}[\phi'(W')] = 0.\]

Hence $\lone{f} \leq \sqrt{|G \times G|} \linfty{\widehat{f}} \leq |G| \alpha$.
\end{proof}

%% file: extractor.tex
\section{The Non-Malleable Extractor}\label{ext}

Our basic extractor was introduced by Chor and Goldreich \cite{cg:weak}.
They showed that it was a two-source extractor for entropy rates bigger than $1/2$.
Dodis and Oliveira \cite{DodO} showed that it was strong.
Neither result implies anything about non-malleability.

To output $m$ bits, we set $M=2^m$ and
choose a prime power $q > M$.
In our basic extractor, we require
that $M|(q-1)$.  Later, we remove this assumption.
Fix a generator $g$ of $\F_q^\times$.
We define $\nm:\F_q^2 \to \dbZ_M$ by $\nm(x,y) = h(\dlog_g(x+y))$.
Here $\dlog_g z$ is the discrete logarithm of~$z$ with respect to $g$, and $h:\dbZ_{q-1} \to \Z_M$ is given by $h(x) = x\mod M$.

In the special case $m=1$, we only require that $q$ is odd.  In this case, $\nm(x,y)$ corresponds to the quadratic character of $x+y$, converted to $\zo$ output.
This is efficient to compute.
Since there is no known efficient deterministic algorithm to find an $n$-bit prime,
we may take $q=3^\ell$, with $3^{\ell - 1} < 2^n < 3^\ell$.

For general~$M$, we use the Pohlig-Hellman algorithm to compute the discrete log mod~$M$.  This runs in polynomial time in the largest prime factor of $M$.
Since in our case $M=2^m$, this is polynomial time.

We still need a prime or prime power $q$ such that $M|(q-1)$.
Unconditionally, we get a polynomial-time algorithm to output $m=c\log n$ bits for any $c>0$.  To output more bits efficiently,
we rely on a widely believed conjecture.
Under Conjecture~\ref{conj-primes}, such a prime can be found efficiently by testing $M+1,2M+1,3M+1,\ldots$
in succession.

Now we prove that $\nm$ is a non-malleable extractor.

\begin{theorem}
\label{basic}
The above function $\nm:\F_q^2 \to \dbZ_M$ is a $(k,\eps)$-non-malleable extractor for $\eps = Mq^{1/4}2^{1-k/2}$. 
\end{theorem}

\begin{proof}
The heart of our proof is a new character sum estimate, given in \theoremref{charsum-uniform} (and
Corollary~\ref{cor:char1}).
We now show how to deduce \theoremref{basic} from the character sum estimate and \lemmaref{special-case}.
Let $X$ be a distribution with $\hinf(X) \geq k$, and let $Y$ be uniform on $\F_q$.
As is well-known, we may assume without loss of generality that $X$ is uniform on a set of size $2^k$.
We set $G=\Z_M$, $(W,W') = (\nm (X,Y), \nm (X,\adv(Y)))$, and we condition on
$Y=y$.

Since $M|(q-1)$,
we have that for $\phi$ a character of $G$, the function $\chi(z) = \phi(h(\dlog_g (z)))$ is a multiplicative character of $\F_q$.  Therefore, Corollary~\ref{cor:char1}
shows that $((W,W')|Y=y)$ satisfies the hypotheses of \lemmaref{special-case} for some $\eta_y$,
where $\expect_{y \leftarrow Y} [\eta_y] \leq \eta$ for $\eta < q^{1/4} 2^{1-k/2} $.
Thus, by \lemmaref{special-case}, $((W,W')|Y=y)$ is $M \eta_y$-close to $((U,h(W'))|Y=y)$ for every $y$.
Since this expression is linear in $\eta_y$, we conclude that $(W,W',Y)$ is $M\eta$-close to $(U,h(W'),Y)$, as required.
\end{proof}

Note that this theorem assumes that the seed is chosen uniformly from $\F_q$,
consistent with Definition~\ref{nmdef}.
However, we may desire to have the seed be a uniformly random bit string.
This causes a problem, since we may not be able to choose $q$ close to
a power of 2.
If we use a $d$-bit seed where $2^d \leq q < 2^{d+1}$,
then we can view the seed as an integer between 0 and $2^d - 1$,
or simply as an element of $\F_q$ with min-entropy at least $(\log q) - 1$.
We can handle this, and in fact much lower min-entropy in the seed, as follows.
First, we
recall the Definition~\ref{def:weak-seed} of a non-malleable extractor with a 
weakly-random seed.
The following lemma shows that a non-malleable extractor with small error remains a non-malleable extractor even
if the seed is weakly random.

\begin{lemma}
\label{weekseed}
A $(k,\eps)$-non-malleable extractor $\nm:[N] \times [D] \to [M]$
is also a $(k,k',\eps')$-non-malleable extractor with $\eps' = (D/2^{k'}) \eps$.
\end{lemma}

\begin{proof}
For $y \in [D]$, let $\eps_y = \Delta((\nm(X,y),\nm(X,\adv(y)),y), (U_{[M]},\nm(X,\adv(y)),y)).$
Then for $Y$ chosen uniformly from $[D]$,
\[
\eps \geq \Delta((\nm(X,Y),\nm(X,\adv(Y)),Y), (U_{[M]},\nm(X,\adv(Y)),Y)) = \frac{1}{D} \sum_{y \in [D]} \eps_y.
\]
Thus, for $Y'$ with $\hinf(Y') \geq k'$, we get
\begin{align*}
\Delta((\nm(X,Y'),\nm(X,\adv(Y')),&Y'), (U_{[M]},\nm(X,\adv(Y')),Y'))\\
&=\sum_{y \in [D]} \Pr[Y=y] \eps_y \leq 2^{-k'}\sum_{y \in [D]} \eps_y \leq (D/2^{k'}) \eps.
\end{align*}
\end{proof}

It is now simple to
analyze our non-malleable extractor as a function $\nm:\zo^n \times \zo^d \to \zo^m$.  Here we work over $\F_q$, where $q$ is the smallest prime
(or prime power) congruent to 1 modulo $M=2^m$.  We let $d=\floor{\log_2 q}$, which is $n+c\log n + O(1)$ under Conjecture~\ref{conj-primes}.
We could even let $d=n$ and the error would only grow by $n^c$.

\begin{theorem}
\label{main}
Under Conjecture~\ref{conj-primes} with constant $c$,
for any $n$, $k > n/2 + (c/2)\log n$, and $m < k/2 - n/4 -(c/4)\log n$,
the above function $\nm:\zo^n \times \zo^d \to \zo^m$ is a polynomial-time computable, $(k,\eps)$-non-malleable extractor for $\eps = O(n^{c/4} 2^{m+n/4-k/2})$.
\end{theorem}

\begin{proof}
Suppose that Conjecture~\ref{conj-primes} holds for the constant $c$.  Then $q=O(n^c 2^n)$, and the seed has min-entropy $k'=d$.  Applying \lemmaref{weekseed},
we obtain error
$$\eps = (q/2^d)Mq^{1/4}2^{1-k/2} = O(n^{c/4} 2^{m+n/4 - k/2}).$$
\end{proof}

After seeing \cite{CRS11}, we improved our character sum to handle weak seeds, using ideas related to their work and \cite{Raz05}.
In particular, we showed \theoremref{charsum}, which implies the following theorem.

\begin{theorem}
\label{main-weakseed}
Under Conjecture~\ref{conj-primes}, for $k \geq (1/2 + \alpha)n$ and $k' \geq (7/\alpha)(m+ \log \eps^{-1}) + 8 \log n$,
the above function is a $(k,k',\eps)$-non-malleable extractor.
\end{theorem}

\begin{proof}
The theorem follows from \theoremref{charsum} in the same way that \theoremref{basic} follows from \theoremref{charsum-uniform}.
\end{proof}

\section{A Character Sum Estimate} \label{csum}

We now prove the necessary character sum estimate.
We prove a somewhat more general statement than is needed for the one-bit extractor, as the general statement is needed to output many bits.
Throughout this section, we take $\F=\F_q$ to be a
finite field with $q$ elements. In addition, we suppose that $\chi:\F^\times
\rightarrow \dbC^\times$ is a nontrivial character of order $d=q-1$, and we extend
 the domain of $\chi$ to $\F$ by taking $\chi(0)=0$.
The following lemma is a consequence of Weil's resolution of the Riemann Hypothesis
for curves over finite fields (see \cite{Wei1948}). In this context, we say that a
polynomial $f\in \F[x]$ {\it has $m$ distinct roots} when $f$ has $m$ distinct roots in
the algebraic closure ${\overline \F}$ of $\F$, or equivalently that such holds in a
splitting field for $f$.

\begin{lemma}\label{weil}
Suppose that $f\in \F[x]$ is a monic polynomial having $m$ distinct roots which is not
a $d$th power in $\F[x]$. Then
$$\Bigl| \sum_{x\in \F}\chi(f(x))\Bigr|\le (m-1)\sqrt{q}.$$
\end{lemma}

\begin{proof} This is immediate from Theorem 2C$^\prime$ of Schmidt
\cite{Sch1976} (see page 43 of the latter source).\end{proof}

 We next consider two arbitrary characters, where the first is nontrivial; without loss of generality we may take
 these to be $\chi_a(x) = (\chi(x))^a$ and $\chi_b(x) = (\chi(x))^b$, where $0 < a < q-1$ and $0 \leq b < q-1$.
Now we establish the main character sum estimate.
Note that we need the assumption that $a \neq 0$:  if $a=0$ and $b = (q-1)/2$, we could take $\adv(y) = 0$ and let $\calS$ be the set of quadratic residues, and then one has no cancellation in the character sum.

\subsection{Character Sum for Uniform Seeds}

We begin by proving the character sum corresponding to uniformly random seeds.
Although this follows from the more general character sum \theoremref{charsum} below, the proof is simpler and gives intuition for the general character sum.  Moreover, this theorem came before \cite{CRS11}, whereas the more general \theoremref{charsum} came afterwards.

\begin{theorem}\label{charsum-uniform}
Suppose that $\calS$ is a non-empty subset of $\F$, and that $\adv:\F\rightarrow \F$
is any function satisfying the property that $\adv(y)\ne y$ for all $y\in \F$. Then
one has
$$\sum_{y\in \F}\Bigl| \sum_{s\in \calS}\chi_a(s+y)\chi_b(s+\adv(y))\Bigr| \le 11^{1/4}q^{5/4}
|\calS|^{1/2}.$$
\end{theorem}

\begin{proof} Write
$\Tet=\sum_{y\in \F}\Bigl| \sum_{s\in \calS}\chi_a(s+y)\chi_b(s+\adv(y))\Bigr|.$
We begin by applying Cauchy's inequality to obtain
$$\Tet^2\le q\sum_{y\in \F}\Bigl|\sum_{s\in \calS}\chi_a(s+y)\chi_b(s+\adv(y))
\Bigr|^2=q\sum_{s,t\in \calS}\sum_{y\in \F}\psi_{s,t}(y),$$
in which we have written
\begin{equation}\label{1-strong}
\psi_{s,t}(y)=\chi_a(s+y)\chi_b(s+\adv(y))\chibar_a(t+y)\chibar_b(t+\adv(y)).
\end{equation}
Applying Cauchy's inequality a second time, we deduce that
$$\Tet^4\le q^2|\calS|^2\sum_{s,t\in \calS}\Bigl|\sum_{y\in \F}\psi_{s,t}(y)
\Bigr|^2.$$
By positivity, the sum over $s$ and $t$ may be extended from $\calS$ to the
entire set $\F$, and thus we deduce that
\begin{equation}\label{2-strong}
\Tet^4\le q^2|\calS|^2\sum_{s,t\in \F}\sum_{y,z\in \F}\psi_{s,t}(y)
\psibar_{s,t}(z).
\end{equation}
On recalling the definition (\ref{1-strong}), we may expand the right hand side of
(\ref{2-strong}) to obtain the bound
\begin{equation}\label{3-strong}
\Tet^4\le q^2|\calS|^2\sum_{y,z\in \F}|\nu(y,z)|^2,
\end{equation}
where
$$\nu(y,z)=\sum_{s\in \F}\chi_a(s+y)\chi_b(s+\adv(y))\chibar_a(s+z)\chibar_b(s+\adv(z)).$$

\par Recall now the hypothesis that $y\ne \adv(y)$. It follows that, considered as
an element of $\F[x]$, the polynomial
$$h_{y,z}(x)=(x+y)^a(x+\adv(y))^b(x+z)^{q-1-a}(x+\adv(z))^{q-1-b}$$
can be a $d$th power only when $y=z$, or when $y=\adv(z)$, $a=b$ and $z=\adv(y)$. In order to confirm
this assertion, observe first that when $y\ne z$ and $y\ne \adv(z)$, then $h_{y,z}$ has a zero of multiplicity $a$ at
$-y$. Next, when $y=\adv(z)$, one has $z\ne y$, and so when $a\ne b$ the polynomial $h_{y,z}$ has a zero of
multiplicity $q-1+a-b$ at $-y$. Finally, when $y=\adv(z)$ and $a=b$, then provided that $z\ne \adv(y)$ one finds
that $h_{y,z}$ has a zero of multiplicity $q-1-a$ at $-z$. In all of these situations it follows that $h_{y,z}$ has a
zero of multiplicity not divisible by $d=q-1$. When $y\ne z$, and $(y,z)\ne (\adv(z),\adv(y))$, therefore,
the polynomial $h_{y,z}(x)$ is not a $d$th power in
$\F[x]$, and has at most $4$ distinct roots. In such a situation, it therefore
follows from \lemmaref{weil} that
$$\nu(y,z)=\sum_{s\in \F}\chi(h_{y,z}(s))$$
is bounded in absolute value by $3\sqrt{q}$. Meanwhile, irrespective of the
values of $y$ and $z$, the expression $\nu(y,z)$ is trivially bounded in
absolute value by $q$. Substituting these estimates into (\ref{3-strong}), we arrive at
 the upper bound
\begin{align*}
\Tet^4&\le q^2|\calS|^2\sum_{y\in \F}\Bigl( |\nu(y,y)|^2+|\nu(y,\adv(y))|^2+
\sum_{z\in \F\setminus\{y,\adv(y)\}}|\nu(y,z)|^2\Bigr)\\
&\le q^2|\calS|^2\sum_{y\in \F}(q^2+q^2+q(3\sqrt{q})^2)=11q^5|\calS|^2.
\end{align*}
We may thus conclude that $\Tet\le 11^{1/4}q^{5/4}|\calS|^{1/2}$.
\end{proof}

A direct computation yields the following corollary.

\begin{corollary}\label{cor:char1}
Under the hypotheses of the statement of \theoremref{charsum}, one has
$$\sum_{y\in \F}\Bigl| \sum_{s\in \calS}\chi_a(s+y)\chi_b(s+\adv(y))\Bigr| < \eta q|\calS|$$
where $\eta < 2q^{1/4}/|\calS|^{1/2}$.
\end{corollary}

\subsection{Character Sum for Weak Seeds}

The work in this subsection came after \cite{CRS11}, and uses ideas related to their work and to \cite{Raz05}.

\begin{theorem}\label{charsum}
For $0<\alpha, \eta \leq 1/2$, suppose
that $\calS$ and $\calT$ are non-empty subsets of $\F$
with $|\calS| \geq q^{1/2 + \alpha}$ and $|\calT| \geq \max((1/\eta)^{7/\alpha}, (\log q)^8)$,
and that $\adv:\F\rightarrow \F$
is any function satisfying the property that $\adv(y)\ne y$ for all $y\in \F$.
Then for large enough~$q$, we have
\[ \sum_{y\in \calT}\Bigl| \sum_{s\in \calS}\chi_a(s+y)\chi_b(s+\adv(y))\Bigr| < \eta |\calT||\calS|. \]
\end{theorem}

We prove this by choosing a suitable parameter $r$ in the following theorem.

\begin{theorem}\label{charsum-param}
Suppose that $\calS$ and $\calT$ are non-empty subsets of $\F$, and that $\adv:\F\rightarrow \F$
is any function satisfying the property that $\adv(y)\ne y$ for all $y\in \F$. Then for each natural
number $r$, one has
$$\sum_{y\in \calT}\Bigl| \sum_{s\in \calS}\chi_a(s+y)\chi_b(s+\adv(y))\Bigr| \le \lam_rq^{1/(4r)}
|\calS|^{1-1/(2r)}|\calT|,$$
where
$$\lam_r=\left( (4r-1)^2+(2r)^{4r}q|\calT|^{-r}\right)^{1/(4r)}.$$
\end{theorem}

\begin{proof} Write
$$\Tet=\sum_{y\in \calT}\Bigl| \sum_{s\in \calS}\chi_a(s+y)\chi_b(s+\adv(y))\Bigr|.$$
We begin by applying Cauchy's inequality to obtain
$$\Tet^2\le |\calT|\sum_{y\in \calT}\Bigl|\sum_{s\in \calS}\chi_a(s+y)\chi_b(s+\adv(y))
\Bigr|^2=|\calT|\sum_{s,t\in \calS}\sum_{y\in \calT}\psi_{s,t}(y),$$
in which we have written
\begin{equation}\label{1}
\psi_{s,t}(y)=\chi_a(s+y)\chi_b(s+\adv(y))\chibar_a(t+y)\chibar_b(t+\adv(y)).
\end{equation}
Applying H\"older's inequality, we deduce that
$$\Tet^{4r}\le |\calT|^{2r}|\calS|^{4r-2}\sum_{s,t\in \calS}\Bigl|\sum_{y\in \calT}\psi_{s,t}(y)
\Bigr|^{2r}.$$
By positivity, the sum over $s$ and $t$ may be extended from $\calS$ to the
entire set $\F$, and thus we deduce that
\begin{equation}\label{2}
\Tet^{4r}\le |\calT|^{2r}|\calS|^{4r-2}\sum_{s,t\in \F}\sum_{{\mathbf y}\in \calT^{2r}}
\prod_{i=1}^r\psi_{s,t}(y_i)\psibar_{s,t}(y_{r+i}).
\end{equation}
On recalling the definition (\ref{1}), we may expand the right hand side of
(\ref{2}) to obtain the bound
\begin{equation}\label{3}
\Tet^{4r}\le |\calT|^{2r}|\calS|^{4r-2}\sum_{{\mathbf y}\in \calT^{2r}}|\nu({\mathbf y})|^2,
\end{equation}
where
$$\nu({\mathbf y})=\sum_{s\in \F}\prod_{i=1}^r\chi_a(s+y_i)\chi_b(s+\adv(y_i))
\chibar_a(s+y_{r+i})\chibar_b(s+\adv(y_{r+i})).$$

\par We now consider the circumstances in which, considered as an element of $\F [x]$,
the polynomial
$$h_{\mathbf y}(x)=\prod_{i=1}^r(x+y_i)^a(x+\adv(y_i))^b(x+y_{r+i})^{q-1-a}(x+\adv(y_{r+i}))^{q-1-b}$$
is a $d$th power. Consider a fixed $2r$-tuple ${\mathbf y}\in \calT^{2r}$ and an index
$i$ with $1\le i\le 2r$. If there is no index $j$ with $1\le j\le 2r$ and $j\ne i$ for which
\begin{equation}\label{w1}
y_i=y_j\quad \text{or}\quad y_i=\adv(y_j),
\end{equation}
then in view of our hypothesis that $y_i\ne \adv(y_i)$, it follows that the polynomial
$h_{\mathbf y}(x)$ has a zero of order precisely $a$ at $-y_i$ in the situation where
$1\le i\le r$, or $q-1-a$ at $-y_i$ in the situation where $r+1\le i\le 2r$. Write $\calB$
 for the set of $2r$-tuples ${\mathbf y}\in \calT^{2r}$ having the property that, for
each index $i$ with $1\le i\le 2r$, there exists an index $j$ with $1\le j\le 2r$ and
$j\ne i$ for which (\ref{w1}) holds. It follows that if $h_{\mathbf y}(x)$ is to be a
$d$th power in $\F[x]$, then one must have ${\mathbf y}\in \calB$. On the other
hand, when ${\mathbf y}\in \calT^{2r}\setminus \calB$, the polynomial
$h_{\mathbf y}(x)$ is not a $d$th power in $\F[x]$, and has at most $4r$ distinct
roots. In such a situation, we therefore deduce from \lemmaref{weil} that
$$\nu({\mathbf y})=\sum_{s\in \F}\chi(h_{\mathbf y}(s))$$
is bounded in absolute value by $(4r-1)\sqrt{q}$. Meanwhile, irrespective of the
value of ${\mathbf y}$, the expression $\nu({\mathbf y})$ is trivially bounded in
absolute value by $q$. Substituting these estimates into (\ref{3}), we arrive at
 the upper bound
\begin{align}
\Tet^{4r}&\le |\calT|^{2r}|\calS|^{4r-2}\Bigl(\sum_{{\mathbf y}\in \calT^{2r}
\setminus \calB}|\nu({\mathbf y})|^2+\sum_{{\mathbf y}\in \calB}
|\nu({\mathbf y})|^2\Bigr)\notag \\
&\le |\calT|^{2r}|\calS|^{4r-2}\left( (4r-1)^2q|\calT|^{2r}+q^2|\calB|\right).\label{w0}
\end{align}

\par It remains now only to bound $|\calB|$. We establish shortly that the
$2r$-tuples ${\mathbf y}$ lying in $\calB$ are generated via the relations (\ref{w1})
from at most $r$ of the coordinates $y_i$ of ${\mathbf y}$. With this in mind, we
begin by bounding the number of $2r$-tuples ${\mathbf y}$ generated from one such
$r$-tuple. Suppose that there are $l$ distinct values amongst $y_1,\ldots ,y_{2r}$, say
$y_{i_1}=v_1,\ldots ,y_{i_l}=v_l$, with respective multiplicities $a_1,\ldots ,a_l$.
Considering a fixed choice of $v_1,\ldots ,v_l$, the number of ways in which
$y_1,\ldots ,y_{2r}$ may be selected to satisfy the multiplicity condition is
$(2r)!/(a_1!a_2!\ldots a_l!)$.

\par Next we identify a directed graph with vertices labelled by the distinct elements
$v_1,\ldots ,v_l$ of $\F$ as follows. We consider the vertices $v_1,v_2,\ldots ,v_l$
in turn. At stage $i$ we consider all elements $v_j$ with $\adv(v_j)=v_i$. If no such
element $v_j$ exists, then we add no edge. If one or more exist, on the other hand,
then we select one such element $v_j$ at random, and add a directed vertex from $v_j$
to $v_i$. Notice that, since $v_1,\ldots ,v_l$ are distinct, it follows that there can be at
most one directed edge leaving any given vertex. Also, by construction, there is at most
one directed edge arriving at any given vertex. Furthermore, in view of the criterion
(\ref{w1}), any vertex $v_k$ which possesses no edges must necessarily have multiplicity
$a_k\ge 2$. In this way, we see that the graph constructed in this manner consists of at
most a union of isolated vertices, non-branching paths of the shape
\begin{equation}\label{w2}
v_{i_1}\rightarrow v_{i_2}\rightarrow \cdots \rightarrow v_{i_k},
\end{equation}
and cycles of the shape
\begin{equation}\label{w3}
v_{i_1}\rightarrow v_{i_2}\rightarrow \cdots \rightarrow v_{i_k}\rightarrow v_{i_1}.
\end{equation}
In the latter two cases, of course, one has $k\ge 2$. For each non-branching path of type
(\ref{w2}), we call the element $v_{i_1}$ the {\it root of the path}. For each cycle of type
(\ref{w3}), we call the element $v_{i_1}$ a {\it root of the cycle}, though of course which
element we label as $v_{i_1}$ is unimportant. Notice that since
$v_{i_{m+1}}=\adv(v_{i_m})$ for each $m<k$, roots uniquely determine all elements
in the repective paths and cycles by repeated application of $\adv$. Consequently, all of the
elements $v_1,\ldots ,v_l$ are uniquely determined by the identities of the roots, and the
indices defining the paths, cycles and isolated vertices of the graph.\par

Denote by $z$ the number of paths and cycles in the graph, and by $w$ the number of
isolated vertices in the graph. Then on considering the multiplicities associated with the
elements $v_1,\ldots ,v_l$, one finds that
$$2z+2w\le a_1+\ldots +a_l=2r.$$
The number of elements from $\calT$ that can occur as roots and isolated vertices is
consequently at most $|\calT|^{z+w}\le |\calT|^r$. We estimate the number of possible
arrangements of indices defining the paths, cycles and isolated vertices as follows. Given
each element $v_i$, we can attach to it a directed path going to another element $v_j$ in
at most $l-1$ ways, or choose not to attach a directed path from it. Thus there are in total
at most $((l-1)+1)^l$ possible arrangements of indices defining the paths, cycles and
isolated vertices amongst the $l$ elements $v_1,\ldots ,v_l$. Combining the estimates that
we have assembled thus far, we conclude that
$$|\calB|\le \sum_{\substack{1\le a_1,\ldots ,a_l\le 2r\\ a_1+\ldots +a_l=2r}}
\frac{(2r)!}{a_1!a_2!\ldots a_l!}l^l|\calT|^r\le (1+\ldots +1)^{2r}l^l|\calT|^r\le
(2r)^{4r}|\calT|^r.$$

\par Finally, on substituting this estimate into (\ref{w0}), we obtain
$$\Tet^{4r}\le |\calT|^{4r}|\calS|^{4r-2}q\left( (4r-1)^2+q(2r)^{4r}|\calT|^{-r}\right),$$
and the conclusion of the theorem follows on extracting $4r$-th roots.
\end{proof}

A direct computation yields the following corollary.

\begin{corollary}
\label{charsum-cor}
Let $\eta$ be a positive number with $\eta\le 1$. Then under
the hypotheses of the statement of \theoremref{charsum-param}, one has
\begin{equation}\label{w4}
\sum_{y\in \calT}\Bigl| \sum_{s\in \calS}\chi_a(s+y)\chi_b(s+\adv(y))\Bigr| < \eta |\calT||\calS|
\end{equation}
whenever $|\calT| \geq (2r)^4q^{1/r}$ and $|\calS| \geq 4rq^{1/2}/\eta^{2r}$.
\end{corollary}

\begin{proof} Recall the notation of the statement of \theoremref{charsum-param}. When
$|\calT|>(2r)^4q^{1/r}$, we find that
$$\lam_r^{4r}=(4r-1)^2+(2r)^{4r}q|\calT|^{-r} \leq (4r-1)^2+1<16r^2.$$
But then the upper bound (\ref{w4}) follows from \theoremref{charsum} provided only that
$$(16r^2)^{1/(4r)}q^{1/(4r)}|\calS|^{1-1/(2r)}|\calT| \leq \eta |\calT||\calS|,$$
as is the case whenever $|\calS| \geq 4rq^{1/2}/\eta^{2r}$.
\end{proof}

\begin{proof}[Proof of \theoremref{charsum}]
We verify that the hypotheses of \theoremref{charsum} imply the conditions on $|\calS|$ and $|\calT|$
in \corollaryref{charsum-cor}, for $r=1+\floor{(2\log q)/\log |\calT|} \geq 3$.  We then have $q^{2/r} \leq |\calT| \leq q^{2/(r-1)}$, and
for large enough $q$ we get $2r \leq \log q$.
Therefore $|\calT| \geq (\log q)^4 |\calT|^{1/2} \geq (2r)^4 q^{1/r}$.

Moreover,
$1/\eta^{2r} \leq |T|^{2 \alpha r/7} \leq q^{(4 \alpha/7)( r/(r-1))} \leq q^{(4 \alpha /7) (3/2)}$, and hence for large enough $q$ we have
\[ |\calS| \geq 4(\log q)q^{1/2+ 6\alpha/7} \geq 4rq^{1/2}/\eta^{2r},\]
as required.
\end{proof}

%% file: privacy.tex
\section{Application to Privacy Amplification}\label{privacy}

Following \cite{kr:agree-close}, we define a privacy amplification protocol $(P_A, P_B)$, executed by two parties Alice and Bob sharing a secret $X\in \bits^n$, in the presence of an active, computationally unbounded adversary Eve, who might have some partial information $E$ about $X$ satisfying $\thinf(X|E)\ge k$.
Informally, this means that whenever a party (Alice or Bob) does not reject, the key $R$ output by this party is random and statistically independent of Eve's view. Moreover, if both parties do not reject, they must output the same keys $R_A=R_B$ with overwhelming probability.

More formally, we assume that Eve is in full control of the communication channel between Alice and Bob, and can arbitrarily insert, delete, reorder or modify messages sent by Alice and Bob to each other. In particular, Eve's strategy $P_E$ actually defines two correlated executions $(P_A,P_E)$ and $(P_E,P_B)$ between Alice and Eve, and Eve and Bob, called ``left execution'' and ``right execution'', respectively. We stress that the message scheduling for both of these executions is completely under Eve's control, and Eve might attempt to execute a run with one party for several rounds before resuming the execution with another party. However, Alice and Bob are assumed to have fresh, private and independent random tapes $Y$ and $W$, respectively, which are not known to Eve (who, by virtue of being unbounded, can be assumed deterministic). At the end of the left execution $(P_A(X,Y),P_E(E))$, Alice outputs a key $R_A\in \bits^m \cup \{\perp\}$, where $\perp$ is a special symbol indicating rejection. Similarly, Bob outputs a key $R_B \in \bits^m \cup \{\perp\}$ at the end of the right execution $(P_E(E),P_B(X,W))$. We let $E'$ denote the final view of Eve, which includes $E$ and the communication transcripts of both executions $(P_A(X,Y),P_E(E))$ and $(P_E(E),P_B(X,W)$. We can now define the security of $(P_A,P_B)$. Our definition is based on \cite{kr:agree-close}.

\BD An interactive protocol $(P_A, P_B)$, executed by Alice and Bob on a communication channel fully controlled by an active adversary Eve, is a $(k, m, \e)$-\emph{privacy amplification protocol} if it satisfies the following properties whenever $\thinf(X|E) \geq k$:
\begin{enumerate}
\item \underline{Correctness.} If Eve is passive, then $\Pr[R_A=R_B \land~ R_A\neq \perp \land~ R_B\neq \perp]=1$.
\item \underline{Robustness.} We start by defining the notion of {\em pre-application} robustness, which states that even if Eve is active, $\Pr[R_A \neq R_B \land~ R_A \neq \perp \land~ R_B \neq \perp]\le \e$.

The stronger notion of {\em post-application} robustness is defined similarly, except Eve is additionally given the key $R_A$ the moment she completed the left execution $(P_A,P_E)$, and the key $R_B$ the moment she completed the right execution $(P_E,P_B)$. For example, if Eve completed the left execution before the right execution, she may try to use $R_A$ to force Bob to output a different key $R_B\not\in\{R_A,\perp\}$, and vice versa.
\item \underline{Extraction.} Given a string $r\in \bits^m\cup \{\perp\}$, let $\purify(r)$ be $\perp$ if $r=\perp$, and otherwise replace $r\neq \perp$ by a fresh $m$-bit random string $U_m$:  $\purify(r)\leftarrow U_m$. Letting $E'$ denote Eve's view of the protocol, we require that
\[\Delta((R_A, E'),(\purify(R_A), E')) \leq \e
~~~~\mbox{and}~~~~
\Delta((R_B, E'),(\purify(R_B), E')) \leq \e\]
Namely, whenever a party does not reject, its key looks like a fresh random string to Eve.
\end{enumerate}
The quantity $k-m$ is called the \emph{entropy loss} and the quantity $\log (1/\e)$ is called the \emph{security parameter} of the protocol.
\ED

\subsection{Case of $k > n/2$}
Given a security parameter $s$, Dodis and Wichs showed that a non-malleable extractor, which extracts at least $2\log n+2s+4$ number of bits with error $\e=2^{-s-2}$, yields a two-round protocol for privacy amplification with optimal entropy loss, which also uses any (regular) extractor $\Ext$ with optimal entropy loss and any asymptotically good one-time message-authentication code $\mac$ (see Definition~\ref{def:mac}), is depicted in Figure~\ref{fig:AKA1}.


\begin{figure}[htb]
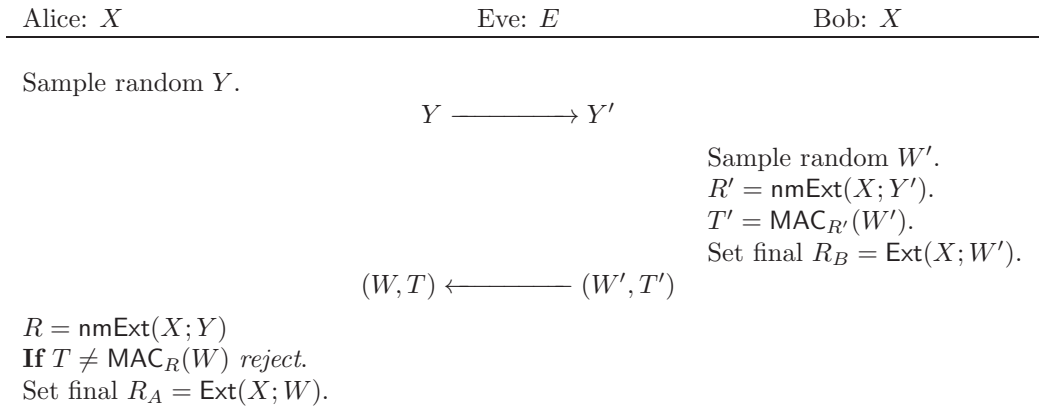

\begin{center}
\begin{small}
\begin{tabular}{l c l}
Alice:  $X$ & Eve: $E$ & ~~~~~~~~~~~~Bob: $X$ \\

\hline\\
Sample random $Y$.& &\\
 & $Y \llrightarrow[\rule{1.5cm}{0cm}]{} Y'$ & \\
&& Sample random $W'$.\\
&&  $R' = \nmExt(X;Y')$.\\
&&  $T' = \mac_{R'}(W')$.\\
&&  Set final $R_B = \Ext(X;W')$.\\
 & $(W,T) \llleftarrow[\rule{1.5cm}{0cm}]{} (W',T')$ & \\
 $R = \nmExt(X;Y)$&&\\
{\bf If} $T \neq \mac_{R}(W)$ {\em reject}.&&\\
Set final $R_A = \Ext(X;W)$.&&\\
\hline
\end{tabular}
\end{small}
\caption{\label{fig:AKA1}
$2$-round Privacy Amplification Protocol for $\thinf(X|E)>n/2$.
}
\end{center}
\end{figure}


Using the bound from Theorem~\ref{main} and setting $\eps=2^{-s}$ and $m=s$, we get the following theorem.
\BT \label{thm:nmext}
Under Conjecture~\ref{conj-primes} with constant $c$, for any $s>0$ there is a polynomial time computable $(k, \e)$-non-malleable extractor with $m=s$ and $\e=2^{-s}$, as long as $k \geq n/2+(c/2)\log n+4s+O(1)$.
\ET

Using this theorem, we obtain the following.

\begin{theorem}\label{thm:2round}
Under Conjecture~\ref{conj-primes} with constant $c$,
there is a polynomial-time, two-round protocol for privacy amplification with security parameter $s$ and entropy loss $O(\log n+s)$, when the min-entropy~$k$ of the $n$-bit secret
satisfies $k \geq n/2 + (c/2+8)\log n + 8s + O(1)$.
\end{theorem}

\paragraph{Using Weak Local Randomness.} We notice that we can use
Theorem~\ref{main-weakseed} to argue that Alice does not need perfect local randomness $Y$ to run the protocol in Figure~\ref{fig:AKA1}. Indeed, since the output of the non-malleable extractor is only $O(s)$-bit long, we only need the min-entropy of $Y$ to be $O(s)$. Similarly, Bob could use a two-source extractor $\Ext$ with a weak seed $W$ constructed by Raz~\cite{Raz05}. Assuming the entropy rate of $X$ is above $1/2+\alpha$ for some $\alpha>0$, this extractor extracts $\Omega(n)$ bits from $X$, and only needed the min-entropy of $W$ to be $O(s)$ as well. To summarize, Alice and Bob can each use local sources of randomness of min-entropy only $O(s)$, and still extract $\Omega(n)$ secret bits from $X$.

\subsection {Case of $k=\delta n$}
Here we give our protocol for arbitrary positive entropy rate. We first give some preliminaries.

\subsubsection{Prerequisites from previous work}


\BD
An {\em elementary somewhere-$k$-source} is a vector of sources $(X_1, \cdots, X_C)$, such that some $X_i$ is a $k$-source. A {\em somewhere $k$-source} is a convex combination of elementary somewhere-$k$-sources.
\ED

\BD
A function $\Cond: \bits^n \to (\bits^{n'})^C$ is a {\em $(k \to k', \e)$-somewhere-condenser} if for every $k$-source $X$, the vector $(X_1,\ldots,X_C) = \Cond(X)$ is $\e$-close to a somewhere-$k'$-source. When convenient, we call $\Cond$ a {\em rate-$(k/n \to k'/n', \e)$-somewhere-condenser}.
\ED

We are going to use condensers recently constructed based on the sum-product theorem. Specifically, we have the following theorem.

\BT [\cite{BarakKSSW05, Raz05, Zuc07}] \label{thm:swcondenser}
For any $\delta>0$ and constant $\beta>0$, there is an efficient family of rate-$(\delta \to 1-\beta, \e=2^{-\Omega(\delta n)})$-somewhere condensers $\Cond: \bits^n \to (\bits^{n'})^C$, where $C=\poly(1/\delta)$ and $n'=\poly(\delta)n$.

\ET






One-time message authentication codes (MACs) use a shared random key to authenticate a message in the information-theoretic setting.
\begin{definition} \label{def:mac}
A function family $\{\mac_R : \bits^{d} \to \bits^{v} \}$ is a $\e$-secure one-time MAC for messages of length $d$ with tags of length $v$ if for any $w \in \bits^{d}$ and any function (adversary) $A : \bits^{v} \to \bits^{d} \times \bits^{v}$,

\[\Pr_R[\mac_R(W')=T' \wedge W' \neq w \mid (W', T')=A(\mac_R(w))] \leq \e,\]
where $R$ is the uniform distribution over the key space $\bits^{\ell}$.
\end{definition}

\begin{theorem} [\cite{kr:agree-close}] \label{thm:mac}
For any message length $d$ and tag length $v$,
there exists an efficient family of $(\lceil  \frac{d}{v} \rceil 2^{-v})$-secure
$\mac$s with key length $\ell=2v$. In particular, this $\mac$ is $\eps$-secure when
$v = \log d + \log (1/\e)$.\\
More generally, this $\mac$ also enjoys the following security guarantee, even if Eve has partial information $E$ about its key $R$.
Let $(R, E)$ be any joint distribution.
Then, for all attackers $A_1$ and $A_2$,
\[\Pr_{(R, E)} [\mac_R(W')=T' \wedge W' \neq W \mid W = A_1(E),~(W', T') = A_2(\mac_R(W), E)] \leq \left \lceil  \frac{d}{v} \right \rceil 2^{v-\thinf(R|E)}.\]
(In the special case when $R\equiv U_{2v}$ and independent of $E$, we get the original bound.)
\end{theorem}

Finally, we will also need to use any strong seeded $(k,\e)$-extractor with optimal entropy loss $O(\log(1/\e))$. A simple extractor that achieves this is the one from the leftover hash lemma, which uses a linear-length seed. We can also use more sophisticated constructions such as those in \cite{GuruswamiUV09, DvirKSS09}, and the non-malleable extractor with short seed length \cite{CRS11} to reduce the communication complexity of the protocol.




\subsubsection{The protocol}

Now we give our privacy amplification protocol for the setting when $\thinf(X|E) = k \ge \delta n$.
We assume that the error $\e$ we seek satisfies $2^{-\Omega(\delta n)}< \e < 1/n$. In the description below, it will be convenient to introduce an ``auxiliary'' security parameter $s$. Eventually, we will set $s=\log(C/\e)+O(1)=\log(1/\e)+O(1)$, so that $O(C)/2^s<\e$, for a sufficiently large $O(C)$ constant related to the number of ``bad'' events we will need to account for. We will need the following building blocks:


\begin{itemize}
\item Let $\Cond:\zo^n \rightarrow (\zo^{n'})^C$ be a rate-$(\delta \to 0.9, 2^{-s})$-somewhere-condenser. Specifically, we will use the one from \theoremref{thm:swcondenser}, where $C=\poly(1/\delta)=O(1)$, $n'=\poly(\delta)n=\Omega(n)$ and
$2^{-s}\gg 2^{-\Omega(\delta n)}$.

\item Let $\nmExt:\bits^{n'}\times \bits^{d'}\rightarrow \bits^{m'}$ be a $(0.9n',2^{-s})$-non-malleable extractor. Specifically, we will use the one from \theoremref{thm:nmext} (which is legal since $0.9n'\gg n'/2 + O(\log n') + 8s +O(1)$)
and set the output length $m' = 4s$ (see the description of $\mac$ below for more on $m'$.)

\item Let $\Ext:\bits^{n}\times \bits^{d}\rightarrow \bits^{m}$ be a $(k',2^{-s})$-extractor with optimal entropy loss $k'-m = O(s)$. Specifically, we will set $k' = k - (7C+11)s= k - O(s)$,
which means that $m = k-O(s)$ as well. We will use the notation $\Ext_{a..b}(X;W)$, where $1\le a\le b\le m$, to denote the sub-string of extracted bits from bit position $a$ to bit position $b$. We assume the seed length $d\le n$ (e.g., by using a universal hash function, but more seed-efficient extractors will work too, reducing the communication complexity).

\item Let $\mac$ be the one-time, $2^{-s}$-secure MAC for $d$-bit messages, whose key length $\ell'=m'$ (the output length of $\nmExt$). Using the construction from \theoremref{thm:mac},
we set the tag length $v' = s + \log d \le 2s$ (since $d\le n \le 1/\e \le 2^s$), which
means that the key length $\ell' = m' = 2v' \le 4s$.

\item Let $\lrmac$ be the another one-time (``leakage-resilient'') MAC for $d$-bit messages, but with tag length $v=2v'\le 4s$ and key length $\ell = 2v \le  8s$. We will later use the second part of \theoremref{thm:mac} to argue good security of this MAC even when $v'$ bits of partial information about its key is leaked to the attacker. To not confuse the two MACs, we will use $Z$ (instead of $R$) to denote the key of $\lrmac$ and $L$ (instead of $T$) to denote the tag of $\lrmac$.
\end{itemize}
Using the above building blocks, the protocol is given in Figure~\ref{fig:AKA2}. To emphasize the presence of Eve, we will use `prime' to denote all the protocol values seen or generated by Bob; e.g., Bob picks $W_1'$, but Alice sees potentially different $W_1$, etc. Also, for any random variable $G$ used in describing our protocol, we use the notation $G=\perp$ to indicate that $G$ was never assigned any value, because the party who was supposed to assign $G$ rejected earlier. The case of final keys $R_A$ and $R_B$ becomes a special case of this convention.


\begin{figure}[htbp]
\begin{center}
\begin{small}
\begin{tabular}{l c l}
Alice:  $X$ & Eve: $E$ & ~~~~~~~~~~~~Bob: $X$ \\

\hline\\
 $(X_1,\ldots X_C) = \Cond(X)$. & \fbox{{\bf Phase $1$}} &  $(X_1,\ldots X_C) = \Cond(X)$.\\
Sample random $Y_1$.& &\\
 & $Y_1 \llrightarrow[\rule{1.5cm}{0cm}]{} Y_1'$ & \\
&& Sample random $W_1'$.\\
&&  $R_1' = \nmExt(X_1;Y_1')$.\\
&&  $T_1' = \mac_{R_1'}(W_1')$.\\
 & $(W_1,T_1) \llleftarrow[\rule{1.5cm}{0cm}]{} (W_1',T_1')$ & \\
 $R_1 = \nmExt(X_1;Y_1)$&&\\
{\bf If} $T_1 \neq \mac_{R_1}(W_1)$ {\em reject}.&&\\
 $Z_1 = \Ext_{s+1..s+\ell}(X;W_1)$.&&\\
&\fbox{{\bf Phases $2$..$C$}}&\\
{\bf For} $i=2$ {\bf to} $C$&  &{\bf For} $i=2$ {\bf to} $C$\\
~~~Sample random $Y_i$.& &~~~Sample random $W_i'$.\\
~~~$S_{i-1} = \Ext_{1..s}(X;W_{i-1})$. & & \\
 & $(S_{i-1},Y_i) \llrightarrow[\rule{1.5cm}{0cm}]{} (S_{i-1}',Y_i')$ & \\
& & ~~~{\bf If} $S_{i-1}'\neq \Ext_{1..s}(X;W_{i-1}')$ {\em reject}.\\
& & ~~~$Z_{i-1}'= \Ext_{s+1..s+\ell}(X;W_{i-1}')$.\\
& & ~~~$L_i'= \lrmac_{Z_{i-1}'}(W_i')$.\\
& & ~~~$R_i' = \nmExt(X_i;Y_i')$.\\
& & ~~~$T_i' = \mac_{R_i'}(W_i')$.\\
 & $(W_i,T_i,L_i) \llleftarrow[\rule{1.5cm}{0cm}]{} (W_i',T_i',L_i')$ & \\
~~~{\bf If} $L_i\neq \lrmac_{Z_{i-1}}(W_i)$ {\em reject}.\\
~~~$R_i = \nmExt(X_i;Y_i)$.&&\\
~~~{\bf If} $T_i \neq \mac_{R_i}(W_i)$ {\em reject}.&&\\
~~~$Z_i = \Ext_{s+1..s+\ell}(X;W_i)$.&&\\
{\bf EndFor}&&{\bf EndFor}\\
&\fbox{{\bf Phase $C+1$}}&\\
Re-assign $Z_C = \Ext_{1..m'}(X;W_C)$.&&$Z_C'= \Ext_{1..m'}(X;W_C')$\\
Sample random $W_{C+1}$.& &\\
$S_C = \mac_{Z_C}(W_{C+1})$ &&\\
 & $(S_C,W_{C+1}) \llrightarrow[\rule{1.5cm}{0cm}]{} (S_C',W_{C+1}')$ & \\
&&{\bf If} $S_C' \neq \mac_{Z_C'}(W_{C+1}')$ {\em reject}.\\
Set final $R_A = \Ext(X;W_{C+1})$.&&Set final $R_B = \Ext(X;W_{C+1}')$.\\
\hline
\end{tabular}
\end{small}
\caption{\label{fig:AKA2}
$(2C+1)$-round Privacy Amplification Protocol for $\thinf(X|E)>\delta n$.
}
\end{center}
\end{figure}


Our protocol proceeds in $C+1$ Phases. During the first $C$ Phases, we run $C$ sequential copies of the two-round protocol for the entropy-rate greater than $1/2$ case (see Figure~\ref{fig:AKA1}), but use the derived secret $X_i$ (output by the somewhere-condenser) instead of $X$ during the $i$-th run. Intuitively, since one of the values $X_i$ is expected to have entropy rate above $1/2$, we hope that the key $Z_i$ extracted in this Phase is secret and uniform. However, there are several complications we must resolve to complete this template into a secure protocol.

The first complication is that Eve might not choose to execute its run with Alice in a ``synchronous'' manner with its execution with Bob. We prevent such behavior of Eve by introducing ``liveness tests'', where after each Phase Alice has to prove that she participated {\em during} that Phase. Such tests were implicit in the original paper of Renner and Wolf~\cite{RW03}, and made explicit by Khanakurthi and Reyzin~\cite{kr:agree-close}. Each liveness test (except for the last one in Phase $C+1$, to be discussed) consists of Bob sending Alice a seed $W_i'$ for the extractor $\Ext$ (which is anyway sent during the $i$-th Phase), and Alice responding with the first $s$ bits of the extracted output. Intuitively, although Eve may choose to maul the extracted seed (which might be possible for all Phases, where the entropy rate of $X_i$ is below $1/2$), Eve cannot predict the correct output without asking Alice {\em something}. And since Bob does uses a new liveness test between every two Phases, this effectively forces Eve to follow a natural ``synchronous'' interleaving between the left and the right executions.

The second complication comes from the fact that after a ``good'' (rate above $1/2$) Phase $i$ is completed, the remaining phases might use low-rate sources $X_{i+1},\ldots,X_C$. Hence, one needs a mechanism to make sure that once a good key is extracted in some {\em a-priori unknown} phase, good keys will be extracted in future phases as well, even if the remaining derived sources $X_i$ have low entropy-rate. This is done by using a second message authentication code $\lrmac$, keyed by a value $Z_{i-1}'$ extracted by Bob in the previous Phase $(i-1)$, to authenticated the seed $W_i'$ sent in Phase $i$. The only subtlety is that Bob still sends the original MAC of $W_i'$, and this MAC might be correlated with the previous extracted key $Z_{i-1}$ (especially if the Phase $i$ uses ``bad-rate'' $X_i$). Luckily, by using the ``leakage-resilient'' property of our second MAC (stated in \theoremref{thm:mac}), and setting the parameters accordingly, we can ensure that $Z_{i-1}'$ has enough entropy to withstand the ``leakage'' of the original MAC of $W_i'$.

The template above already ensures the {\em robustness} of the protocol, if we were to extract the key $Z_C$ (or $Z_C'$ for Bob) derived at the end of Phase $C$. Unfortunately, it does not necessarily ensure that Alice outputs a {\em random} key (i.e., it does not guarantee the extraction property for Alice). Specifically, by making Alice's execution run faster than Bob's execution, it might be possible for Eve to make Alice successfully accept a non-random seed $W_C$, resulting in non-random key $Z_C$.
Intuitively, since all the $X_i$'s except for one might have low entropy rate, our only hope to argue security should come from the non-malleability on $\nmExt$ in the ``good'' Phase $i$. However, since Bob is behind (say, at Phase $j<i$) Alice during the good Phase $i$,
Bob will use a wrong source $X_j$ for the non-malleable extractor, and we cannot use the non-malleability of $\nmExt$ to argue that Eve cannot fool Alice into accepting a wrong seed $W_i$ (and, then, wrong $W_{i+1},\ldots,W_C$). Of course, in this case we know Bob will eventually reject, since Eve won't be able to answer the remaining liveness tests. However, Alice's key $Z_C$ is still non-random, violating extraction.

This is the reason for introducing the last Phase $C+1$. During this phase Alice (rather than Bob) picks the last seed $W_{C+1}$ and uses it to extract her the final key $R_A$. Therefore, $R_A$ is now guaranteed to be random. However, now we need to show how to preserve robustness and Bob's extraction. This is done by Alice sending the MAC of $W_{C+1}$ using they key $Z_C$ she extracted during the previous round. (We call this MAC $S_C$ rather than $T_{C+1}$, since it also serves as a liveness test for Alice during Phase $(C+1)$.) From the previous discussion, we know that, with high probability, (a) either $Z_C$ is non-random from Eve's perspective, but then Bob will almost certainly reject (ensuring robustness and preserving Bob's extraction); or (b) $Z_C=Z_C'$ is random and secret from Eve, in which case the standard MAC security suffices to ensure both robustness and Bob's extraction.

We detail the formal proof following the above intuition in the next section, which also establishes the desired parameters promised by Theorem~\ref{thm:privacy}.

\subsubsection{Security Proof of Our Protocol (Proof of Theorem~\ref{thm:privacy})}

We start
by noticing that our protocol takes $2C+1 = \poly(1/\delta) = O(1)$ rounds and achieves entropy loss $k-m = O(Cs) = O(\poly(1/\delta)\log(1/\e))$, as claimed. Also, the protocol obviously satisfies the correctness requirement.

We will also assume that the side information $E$ is empty (or fixed to a constant), since by Lemma~\ref{entropies}, with probability $1-2^{-s}$, $\hinf(X|E=e)\ge \delta n - s$, which will not affect any of our bounds. Before proving robustness and extraction properties of our protocol, we start with the following simple observation.

\BL\label{lem:counting}
Let $E'$ be Eve's view at the end of her attack (without the keys $R_A$ and $R_B$ used in the post-application robustness experiment). Then, for any deterministic functions $f$ and $g$, we have
$$\thinf(f(X)~|~g(E')) \ge \hinf(f(X)) - (7C-3)s$$
In particular, recalling that $k' = \hinf(X) - (7C+11)s$, we have $\thinf(X|g(E')) \ge k'+14s$.
\EL
\begin{proof}
Clearly, if it sufficient to prove the claim for $g$ being identity, as it gives the predictor
the most information to guess $f(X)$. Also notice that, at best, if neither party rejects, Eve's view $E'=(\vec{Y},\vec{S},\vec{W'},\vec{T'},\vec{L'},W_{C+1})$, where $\vec{Y} = \{Y_1,\ldots,Y_C\}$,
$\vec{S} = \{S_1,\ldots,S_C\}$, $\vec{W'} = \{W_1',\ldots,W_C'\}$, $\vec{T'} = \{T_1',\ldots,T_C'\}$ and $\vec{L'} = \{L_2',\ldots,L_C'\}$. Since $\vec{Y}$, $\vec{W'}$ and $W_{C+1}$ are independent of $X$ (and, thus, $f(X)$), using Lemma~\ref{lem:amentropy} and recalling $|S_i|=s$ for $i< C$, $|S_C|=|T_i'|=v' \le 2s$, $|L_i'|=v\le 4s$, we have
\begin{eqnarray*}
\thinf(f(X)|E') &\ge& \thinf(f(X)|(\vec{Y},\vec{W'},W_{C+1})) - |\vec{S}| - |\vec{T'}| - |\vec{L'}|\\
&=& \hinf(f(X)) - (C-1)s - v' - Cv' - (C-1)v \\
&\ge& \hinf(f(X)) - (C-1)s - 2(C+1)s - (C-1)4s\\
&=& \hinf(f(X)) - (7C-3)s
\end{eqnarray*}
\end{proof}

Next, we will argue the extraction property for Alice.

\BL\label{lem:robustA}
$$\Delta((R_A,E'),(\purify(R_A),E'))\le 2^{-s+1}$$
\EL
\begin{proof}
Since $\purify(R_A) = R_A$ when Alice rejects (i.e., $R_A = \perp$), it is sufficient to show that Alice's key is close to uniform conditioned on Alice not rejecting, i.e.
\begin{equation}\label{eq:extA}
\Delta((\Ext(X;W_{C+1}),E'),(U_m,E'))\le 2^{-s+1}
\end{equation}

By Lemma~\ref{lem:counting}, $\thinf(X|E') \ge k'+14s$.
Using Lemma~\ref{entropies}, we get that $\Pr_{e'\leftarrow E'}[\hinf(X|E'=e')\ge k'] \ge 1-2^{-s}$.
Since $\Ext$ is $(k',2^{-s})$-extractor,
Equation~(\ref{eq:extA}) immediately follows the triangle inequality and the security of the extractor, by conditioning on whether or not $\hinf(X|E'=e')\ge k'$.
\end{proof}

Next, we notice that in order to violate either robustness of Bob's extraction, Eve must make Bob accept (i.e., $R_B\neq \perp$). Therefore, we start by examining how Eve might cause Bob to accept. Notice, since Alice sends $C+1$ messages, including the first and the last message, Eve can make $C+1$ calls to Alice, which we call $Alice_1,\ldots,Alice_{C+1}$, where, for each call $Alice_i$,
$1\le i\le C+1$, Eve gets back the message sent by Alice during Phase $i$. Additionally, Alice also computes her key $R_A$ in response to $Alice_{C+1}$ (and gives $R_A$ to Eve, in addition to $S_C$ and $W_{C+1}$, for post-application robustness). Similarly, Eve can also make $C+1$ calls to Bob, denoted $Bob_1,\ldots,Bob_{C+1}$, where each call $Bob_i$ expects as input the message that
Alice supposedly sent to Bob in Phase $i$. When $i\le C$, Bob responds to such a message with his own message in Phase $i$. When $i=C+1$, Bob computes his key $R_B$ (and gives $R_B$ to Eve for post-application robustness).
Clearly, the $(C+1)$ calls to Alice must be made in order, and the same the $(C+1)$ calls to Bob. However, a malicious Eve might attempt to interleave the calls in some adversarial manner to make Bob accept. We say that Eve is {\em synchronous} if he makes his oracle calls in the (``synchronous'') order $Alice_1,Bob_1,Alice_2,Bob_2,\ldots,Alice_{C+1},Bob_{C+1}$. We notice that, without loss of generality, Eve always starts by making the $Alice_1()$ call, since this call has no inputs Eve needs to provide. Namely, Eve must as well find out the values $Y_1$ first, and, if she wants, delay using this value until later. With this convention in mind, we show that Eve {\em must be synchronous in order to make Bob accept}.

\BL\label{lem:synchrony}
\begin{equation}\label{eq:robust}
\Pr[R_B\neq \perp \land \mbox{\rm{~Eve~is~not~synchronous}}]\le \frac{3C}{2^s}
\end{equation}
\EL
\begin{proof}
As we said, we assume Eve always makes the call $Alice_1$ first. After that, Eve makes $C+1$ calls to Bob and $C$ calls to Alice in some order. We claim that for every $1\le i\le C$, Eve must make  at least one call to some $Alice_j$ in between two successive calls $Bob_i$ and $Bob_{i+1}$. If we show this (with total failure probability from Equation~(\ref{eq:robust})), Eve must be synchronous, since the synchronous scheduling is the only scheduling that starts with $Alice_1$ and has a fresh call to Alice between $Bob_1$ and $Bob_2$, $Bob_2$ and $Bob_3$, $\ldots$, $Bob_{C}$ and $Bob_{C+1}$.

Given $1\le i\le C$, let $F_i$ denote the event that Eve's scheduling of calls made two successive calls $Bob_i$ and $Bob_{i+1}$ without a fresh call to some $Alice_j$, and Bob does not reject after the call $Bob_{i+1}$. We claim that $\Pr[F_i]\le 3/2^s$. The bound from Equation~(\ref{eq:robust}) then follows by simply taking the union bound over all $i$. We consider two cases:

{\bf Case 1: $1\le i<C$}. In this case, after the call $Bob_i(\cdot,\cdot)$ is made, Bob picks a
fresh seed $W_i'$, and returns it as part of the output. By assumption, Eve immediately makes a call $Bob_{i+1}(S_i',\cdot)$, without any intermediate calls to Alice, and Bob rejects if $S_i'\neq \Ext_{1\ldots s}(X;W_i')$. Thus, to establish our claim it is sufficient to show that
$\Pr[S_i'\neq \Ext_{1\ldots s}(X;W_i')]\le 3/2^s$. Intuitively, the bound on $\Pr[F_i]$ now follows from the fact that $\Ext$ is a good (strong) $(k',2^{-s})$-extractor, since, conditioned on Eve's information so far,
the $s$-bit value $\Ext_{1\ldots s}(X;W_i')$ is $2^{-s}$-close to random, and, hence, cannot be predicted with probability better that $2^{-s}+2^{-s}$ (the third $2^{-s}$ is due to Lemma~\ref{entropies}, since our extractor is assumed to be worst case, and is not needed for universal hash function extractors~\cite{dors}).

A bit more formally, let $E_i$ be Eve's view before
the call to $Bob_i$ is made, and $E_i' = (E_i,W_i',T_i',L_i')$ be Eve's view after
the call to $Bob_i$ is made. We notice that $E_i'$ is a deterministic function of
$E_i^* = (E_i,Z_{i-1}',R_i')$ and $W_i'$, since $L_i'=\lrmac_{Z_{i-1}'}(W_i')$ and
$T_i' = \mac_{R_i'}(W_i)$. Moreover, $W_i'$ is freshly chosen even conditioned on $E_i^*$.
Thus, $\Pr[F_i]\le \Pr[Eve(E_i^*,W_i') = \Ext_{1..s}(X;W_i')]$, where $W_i'$ is independent of $(X,E_i^*)$. We also note that $\thinf(X|E_i)) \ge k'+14s$, by Lemma~\ref{lem:counting}, since $E_i$ is a function of $E'$. Thus, $\thinf(X|E_i^*)\ge \thinf(X|E_i) - |Z_{i-1}'| - |R_i'| \ge k'+14s - 4s-8s = k'+2s$. Using  Lemma~\ref{entropies},
$\Pr_{e_i^*}[\hinf(X|E_i^*=e_i^*)\ge k']\ge 1-2^{-s}$, and the rest follows from the fact that in this case $(W_i',\Ext_{1..s}(X;W_i'))$ is $2^{-s}$-close to $(W_i',U_s)$, as mentioned earlier.

{\bf Case 2: $i=C$}. In this case, after the call $Bob_C(\cdot,\cdot)$ is made, Bob picks a
fresh seed $W_C'$, and returns it as part of the output. By assumption, Eve immediately makes a call $Bob_{i+1}(S_C',W_{C+1}')$, without any intermediate calls to Alice, and Bob rejects if
$S_C'\neq \mac_{Z_C'}(W_{C+1}')$, where $Z_C' = \Ext_{1\ldots m'}(X;W_i')$.
Thus, to establish our claim it is sufficient to show that
$\Pr[S_C'\neq \mac_{Z_C'}(W_{C+1}')]\le 3/2^s$. Completely similar to the previous case, we can argue that the value $Z_C'$ used by Bob is $2^{1-s}$-close to $U_{m'}$ conditioned on Eve's view so far. Moreover, the $2^{-s}$-security of $\mac$ ensures that, when the key $Z_C'$ is truly uniform, Eve cannot successfully forge a valid tag $\mac_{Z_C'}(W_{C+1}')$ of any (even adversarially chosen) message $W_{C+1}'$ with probability greater than $2^{-s}$, completing the proof of this case as well.
\end{proof}

Therefore, from now on {\em we assume that Eve is indeed synchronous}. Moreover, since Eve must make Bob accept, we assume Eve finishes the both left and right execution (with the last call to $Bob_{C+1}$, hoping that Bob will accept). Also, by \theoremref{thm:swcondenser}, we have that $(X_1, \cdots, X_C)$ is $2^{-\Omega(\delta n)}$-close to a somewhere rate-$0.9$ source. Thus, we will ignore the error and think of $(X_1, \cdots, X_C)$ as indeed being a somewhere rate-$0.9$ source, as it only adds $2^{-\Omega(\delta n)}\ll 2^{-s}$ to the total probability of error. Also, it is sufficient to show robustness and extraction for Bob properties assuming that  $(X_1, \cdots, X_C)$ is an {\em elementary} somewhere rate-$0.9$ source, since $(X_1, \cdots, X_C)$ is a convex combination of elementary somewhere rate-$0.9$ sources. Hence, from now one we assume that some ``good'' index $1\le g\le C$ satisfies $\hinf(X_g) \ge 0.9n'$. We stress that this index $g$ is not known to Alice and Bob, but could be known to Eve. We start by showing that, with high probability, Eve must forward a correct seed $W_g=W_g'$ to Alice in the ``good'' Phase $g$.

\BL \label{lem:good-round}
Assuming Eve is synchronous,
\begin{equation}\label{eq:good-round}
\Pr[R_B\neq \perp \land~ W_g\neq W_g']\le \frac{3}{2^s}
\end{equation}
\EL
\begin{proof}
Let $E_{g-1}'$
be Eve's view before the call to $Alice_g$.
Note that $X_g$ is a deterministic function of $X$, and $(E_{g-1}',S_{g-1},L_g')$ is a deterministic function of Eve's transcript $E'$. Thus, by Lemma~\ref{lem:counting},
\begin{eqnarray*}
\thinf(X_g|(E_{g-1}',S_{g-1},L_g'))&\ge& \hinf(X_g) - (7C-3)s\\
&\ge& 0.9n' - (7C-3)s\\
&=& (n'/2 +O(\log n')+8s+O(1)) + s - (0.4n'-O(Cs+\log n))\\
&\ge& (n'/2 +O(\log n')+8s+O(1)) + s
\end{eqnarray*}
where the last inequality follows since $n' = \poly(1/\delta) n \gg O(Cs+\log n))$.
By \lemmaref{entropies}, with probability $1-2^{-s}$ over the fixings of
$E_{g-1}',S_{g-1},L_g'$,
the min-entropy of $X_g$ conditioned on these fixings is at least $n'/2+O(\log n')+8s+O(1)$.
Notice also that the seed $Y_g$ is independent of $E_{g-1}',S_{g-1},L_g'$.
Moreover, for the argument in this lemma, we will ``prematurely'' give Eve the value $L_g'$ already after the call to $Alice_{g}$ (instead of waiting to get it from the call to $Bob_g$). Let us now summarize the resulting task of Eve in order to make $W_g\neq W_g'$, and argue that Eve is unlikely to succeed.

After the call to $Alice_g$, with high probability the min-entropy of $X_g$ conditioned on Eve's view is greater than $n'/2+O(\log n')+8s+O(1)$, so that we can apply the non-malleability guarantees of $\nmExt$ given by \theoremref{thm:nmext}. Alice then picks a random seed $Y_g$ for $\nmExt$ and gives it to Eve. (Synchronous) Eve then forwards some related seed $Y_g'$ to $Bob_g$ (and another value $S_{g-1}'$ that we ignore here), and learns some message $W_g'$ and the tag $T_g'$ of $W_g'$, under key $R_g' = \nmExt(X_g;Y_g')$ (recall, we assume Eve already knows $L_g'$ from before).
To win the game, Eve must produce a value $W_g\neq W_g'$ and a valid tag $T_g$ of $W_g$ under the original key $R_g=\nmExt(X_g;Y_g)$.

We consider two cases. First, if Eve sets $Y_g'=Y_g$, then $R_g=R'_g$ is $2^{-s}$-close to uniform by \theoremref{thm:nmext}. Now, if $R_g$ was truly uniform, by the one-time unforgeability of $\mac$, the probability that Eve can produce a valid tag $T_g$ of a new message $W_g\neq W_g'$ is at most $2^{-s}$. Hence, Eve cannot succeed with probability more that $2^{-s+1}$ even with $R_g$ which is $2^{-s}$-close to uniform, implying the bound stated in the lemma (since we also lost $2^{-s}$ by using \lemmaref{entropies} at the beginning).

On the other hand, if Eve makes $Y_g'\neq Y_g$, \theoremref{thm:nmext} implies that
$\Delta((R_g, R_g'), (U_{m'}, R_g'))\le 2^{-s}$. Thus, the tag $T_g'$ under $R_g'$ is almost completely useless in predicting the tag of $W_g$ under (nearly random) $R_g$. Therefore, by $2^{-s}$ security of $\mac$, once again the probability that Eve can successfully change $W_g'$ without being detected is at most $2^{-s+1}$ (giving again the final bound $3/2^s$).
\end{proof}

Now, we want to show that, once Eve forwards correct $W_g=W_g'$ to Alice in Phase $g$, Eve must forward correct seeds $W_i=W_i'$ in all future phases $i=g+1,\ldots,C$. We start by the following observation, which states that the derived keys $Z_{i-1}'$ used by Bob in $\lrmac$ look random to Eve {\em whenever Eve forwards a correct key $W_{i-1}=W_{i-1}'$} to Alice.

\BL\label{lem:correctW}
Assume Eve is synchronous, $2\le i\le C$, and Eve forwards a correct value $W_{i-1}=W_{i-1}'$ to Alice during her call to $Alice_{i}$. Also, let $E_i$ be Eve's view after the call to $Alice_i(W_{i-1},\cdot,\cdot)$. Then
\begin{equation}\label{eq:randomZ}
\Delta((Z_{i-1}',E_i),(U_{\ell},E_i))\le \frac{3}{2^s}
\end{equation}
\EL
\begin{proof}
Notice that $E_i = (E_{i-1}, W_{i-1}', T_{i-1}', L_{i-1}', S_{i-1}, Y_i)$, where $E_{i-1}$ is Eve's view after the call to $Alice_{i-1}$. For convenience, we replace the two tags $T_{i-1}', L_{i-1}'$ of $W_{i-1}'$ by the corresponding MAC keys $R_{i-1}', Z_{i-2}'$, respectively, since this gives Eve only more information. Also, since $W_{i-1}=W_{i-1}'$, we know that the value $S_{i-1} = \Ext_{1..s}(X;W_{i-1}) = \Ext_{1..s}(X;W_{i-1}')$. Recalling that $Z_{i-1}' = \Ext_{s+1..s+\ell}(X;W_{i-1}')$, and denoting ``side information'' by $E_i^*=(E_{i-1}, R_{i-1}', Z_{i-2}',Y_i)$, it is enough to argue
\begin{equation}\label{eq:random1}
\Delta((E_i^*, W_{i-1}', \Ext_{1..s}(X;W_{i-1}'), \Ext_{s+1..s+\ell}(X;W_{i-1}'))~,~
       (E_i^*, W_{i-1}', \Ext_{1..s}(X;W_{i-1}'), U_{\ell}))\le \frac{3}{2^s}
\end{equation}
where we notice that $E_i^*$ is {\em independent} of the choice of random $W_{i-1}'$.
In turn, Equation~(\ref{eq:random1}) follows from the fact that $\Ext$ is $(k',2^{-s})$-extractor provided we can show that $\thinf(X|E_i^*)\ge k+s$. Indeed, the first error term $2^{-s}$ comes from Lemma~\ref{entropies} to argue that $\Pr_{e_i^*}[\hinf(X|E_i^*=e_i^*)\ge k]\ge 1-2^{-s}$, and the other two error terms follow from the triangle inequality and the security of the extractor
(first time applies on the first $s$ extracted bits, and then on all $s+\ell$ extracted bits).

So we show that $\thinf(X|E_i^*)\ge k+s$.
\begin{eqnarray*}
\thinf(X|E_i^*) &=& \thinf(X|E_{i-1}, R_{i-1}', Z_{i-2}',Y_i)\\
&\ge& \thinf(X|E_{i-1},Y_i) - |R_{i-1}'| - |Z_{i-2}'|\\
&=& \thinf(X|E_{i-1}) - m' - \ell\\
&\ge& k' + 14s - 4s - 8s\\
&=& k'+2s
\end{eqnarray*}
where the first inequality used Lemma~\ref{lem:amentropy}, the second equality used the fact that $Y_i$ is independent of  $(X,E_{i-1})$, and the second inequality used Lemma~\ref{lem:counting}, since $E_{i-1}$ is deterministic function of $E'$.
\end{proof}

Next, we use Lemma~\ref{lem:good-round} and Lemma~\ref{lem:correctW} to show that, with high probability, Alice and Bob must agree on the same key $Z_C=Z_C'$ when they reach the last Phase $(C+1)$.
\BL \label{lem:final-Z}
Assuming Eve is synchronous,
\begin{equation}\label{eq:last-round}
\Pr[R_B\neq \perp \land~ Z_C\neq Z_C']\le \frac{4C}{2^s}
\end{equation}
\EL
\begin{proof}
Since $Z_C=\Ext_{1\ldots m'}(X;W_C)$ and $Z_C'=\Ext_{1\ldots m'}(X;W_C')$, we get
\begin{eqnarray*}
\Pr[R_B\neq \perp \land~ Z_C\neq Z_C']&\le& \Pr[R_B\neq \perp \land W_C\neq W_C']\\
&\le& \Pr[R_B\neq \perp \land ~W_g\neq W_g'] + \sum_{i=g+1}^C
\Pr[R_B\neq \perp \land~ W_{i-1}=W_{i-1}' \land~ W_i\neq W_i']\\
&\le& \frac{3}{2^s} + (C-1)\cdot \max_{i>g} Pr[R_B\neq \perp \land~ W_{i-1}=W_{i-1}' \land~ W_i\neq W_i']
\end{eqnarray*}
where the second inequality states that in order for $W_C\neq W_C'$, either we must already have $W_g\neq W_g'$ (which, by Lemma~\ref{lem:good-round}, happens with probability at most $3/2^s$), or there must be some initial Phase $i>g$ where $W_{i-1}=W_{i-1}'$ still, but $W_i\neq W_i'$.
Thus, to establish Equation~(\ref{eq:last-round}), it suffices to show that, for any Phase $g<i\le C$,
\begin{equation}\label{eq:induction}
Pr[R_B\neq \perp\land ~ W_{i-1}=W_{i-1}' \land~ W_i\neq W_i']\le \frac{4}{2^s}
\end{equation}

Intuitively, this property follows from the unforgeability of $\lrmac$, since Eve must be able to forge a valid tag $L_i$ of $W_i\neq W_i'$, given a valid tag of $W_i'$ (under the same $Z_{i-1}=Z_{i-1}'$ since $W_{i-1}=W_{i-1}'$). The subtlety comes from the fact that Eve also learns the $v'$-bit value $T_i' = \mac_{R_i'}(W_i')$, which could conceivably be correlated with the key $Z_{i-1}'$ for $\lrmac$. Luckily, since the tag length $v$ of $\lrmac$ is twice as large as $v'$, \theoremref{thm:mac} states that $\lrmac$ is still unforgeable despite this potential ``key leakage''.

More formally, if Eve forwards a correct value $W_{i-1}=W_{i-1}'$, both Alice and Bob use the same key $Z_{i-1}'=Z_{i-1}=\Ext_{s+1..s+\ell}(X;W_{i-1}')$ to $\lrmac$ during Phase $i$.
Moreover, by Lemma~\ref{lem:correctW}, we know that this key $Z_{i-1}$ looks random to Eve right before the call to $Bob_i$: $\Delta((Z_{i-1}',E_i),(U_{\ell},E_i))\le \frac{3}{2^s}$, where
$E_i$ is Eve's view after the call to $Alice_i(W_{i-1},\cdot,\cdot)$.
After the call to $Bob_i$, Eve learns the tag $L_i'$ of $W_i'$, and also a $v'$-bit value $T'$, which, for all we know, might be correlated with the key $Z_{i-1}'$. Therefore, to argue the bound in Equation~(\ref{eq:induction}), it is sufficient to argue that Eve can succeed with probability at most $2^{-s}$ in the following ``truncated'' experiment. After the call to $Alice_i$, the actual key $Z_{i-1}'$ is replaced by uniform $Z_{i-1}^* \leftarrow U_{\ell}$. Then a random message $W_i'$ is chosen, its tag $L_i'$ is given to Eve, and Eve is also allowed to obtain
arbitrary $v'$ bits of information about $Z_{i-1}^*$. Eve succeeds if she can produce a valid tag $L_i$ (under $Z_{i-1}^*$) of a different message $W_i\neq W_i'$. This is precisely the precondition of the second part of \theoremref{thm:mac}, where $\thinf(Z_{i-1}^*|E)\ge \ell - v' = 2v - v/2 = 3v/2$. Hence, Eve's probability of success is at most $d2^{v-3v/2} = d2^{-v/2} =d2^{-v'} \le 2^{-s}$.
\end{proof}

We need one more observation before we can finally argue Bob's extraction and robustness.
Namely, at the end of Phase $C$, (synchronous) Eve has almost no information about the authentication key $Z_C'$ used by the Bob (and Alice, by Lemma~\ref{lem:final-Z}) in the
final Phase $C+1$.

\BL\label{lem:ZC}
Assume Eve is synchronous, and let $E_C'$ be Eve's view after the call to $Bob_C$. Then
\begin{equation}\label{eq:randomZC}
\Delta((Z_C',E_C' \mid R_B\neq \perp),(U_{m'},E_C'\mid R_B\neq \perp))\le \frac{2}{2^s}
\end{equation}
Additionally, $\thinf(X|(E_C',Z_C'))\ge k'+10s$.
\EL
\begin{proof}
The proof is similar to, but simpler than, the proof of Lemma~\ref{lem:correctW}. We notice that $E_C' = (E_C, W_C', T_{C}', L_{C}')$, where $E_C$ is Eve's view after the call to $Alice_C$. For convenience, we replace the two tags $T_{C}', L_{C}'$ of $W_C'$ by the corresponding MAC keys $R_{C}', Z_{C-1}'$, respectively, since this gives Eve only more information.
Recalling that $Z_{C}' = \Ext_{1..m'}(X;W_{C}')$, and denoting ``side information'' by $E_C^*=(E_{C}, R_{C}', Z_{C-1}')$, it is enough to argue
\begin{equation}\label{eq:random2}
\Delta((E_C^*, W_{C}', \Ext_{1..m'}(X;W_{C}'))~,~
       (E_C^*, W_{C}', U_{m'}))\le \frac{2}{2^s}
\end{equation}
where we notice that $E_C^*$ is {\em independent} of the choice of random $W_C'$.
In turn, Equation~(\ref{eq:random2}) follows from the fact that $\Ext$ is $(k',2^{-s})$-extractor provided we can show that $\thinf(X|E_C^*)\ge k+s$, where the extra error term $2^{-s}$ comes from Lemma~\ref{entropies} to argue that $\Pr_{e_C^*}[\hinf(X|E_C^*=e_C^*)\ge k]\ge 1-2^{-s}$.

So we show that $\thinf(X|E_C^*)\ge k+s$.
\begin{eqnarray*}
\thinf(X|E_C^*) &=& \thinf(X|E_{C}, R_{C}', Z_{C-1}')\\
&\ge& \thinf(X|E_{C}) - |R_{C}'| - |Z_{C-2}'|\\
&=& \thinf(X|E_C) - m' - \ell\\
&\ge& k' + 14s - 4s - 8s\\
&=& k'+2s
\end{eqnarray*}
where the first inequality used Lemma~\ref{lem:amentropy}, and the second inequality used Lemma~\ref{lem:counting}, since $E_{C}$ is deterministic function of $E'$.

The final claim $\thinf(X|(E_C',Z_C'))\ge k'+10s$ follows from Lemma~\ref{lem:amentropy} and fact that $\thinf(X|E_C')\ge k'+14s$ (Lemma~\ref{lem:counting}) and $|Z_C'|=m'\le 4s$.
\end{proof}

Lemma~\ref{lem:final-Z} and Lemma~\ref{lem:ZC} imply that, in order for the synchronous Eve to have a non-trivial chance to make Bob accept, at the end of Phase $C$ Alice and Bob must agree on a key $Z_C=Z_C'$ which looks random to Eve. Moreover, $X$ still has a lot of entropy given $Z_C'$ and Eve's view so far. Thus, to show both (post-application) robustness and extraction for Bob, it is sufficient to show these properties for a very simply one-round key agreement protocol, which emulates the final Phase $(C+1)$ of our protocol with Alice and Bob sharing a key $Z_C=Z_C'$ which is assumed to be random and independent from Eve's view so far. We start with post-application robustness.

{\bf Post-Application Robustness:} To cause Bob output a different key than Alice in Phase $(C+1)$, Eve must modify Alice seed $W_{C+1}$ to $W_{C+1}'\neq W_{C+1}$, and then forge a valid tag $S_C'$ of $W_{C+1}'$ under the shared key $Z_C=Z_C'$. For pre-application robustness, the unforgeability of $\mac$ immediately implies that Eve's probability of success is at most $2^{-s}$. However, in the post-application robustness experiment, Eve is additionally given Alice's final key
$R_A=\Ext(X;W_{C+1})$. Luckily, since $X$ has more than $k'+s$ bits of min-entropy {\em even conditioned of the MAC key $Z_C$}, security of the extractor implies that that the joint distribution of $Z_C$ and $R_A$ looks like a pair of independent random strings. In particular, Eve still cannot change the value of the seed $W_{C+1}$ in Phase $(C+1)$, despite being additionally given Alice's key $R_A$, since that key looks random and independent of the MAC key $Z_C=Z_C'$.

{\bf Extraction for Bob:} We just argued (pre-application) robustness of our protocol, which --- for synchronous Eve --- means that if Bob does not reject, then, with high probability, he outputs the same key $R_B = \Ext(X;W_{C+1}')$ as Alice's key $R_A = \Ext(X;W_{C+1})$. Thus, Bob's extraction is implied by Alice's extraction, which was already argued in Lemma~\ref{lem:robustA}. Alternatively, Alice's extraction can be seen directly, as she chooses a fresh seed $W_{C+1}$ and
$\thinf(X|E_C',Z_C)\ge k'+10s$.

This concludes the proof of Theorem~\ref{thm:privacy}.

%% file: appendix.tex
\section{Generalizing the Non-Malleable Extractor}\label{gen}

We now generalize our earlier results to show that we get a non-malleable extractor even if $M$ does not divide $q-1$.
We still use the same function $\nm(x,y) = h(\dlog_g(x+y))$, with $h:\dbZ_{q-1} \to \Z_M$ given by $h(x) = x\mod M$.

\begin{theorem}
\label{arbitraryM}
There exists a constant $c>0$ such that for any $n$, $k>n/2 + \log n + c$, and $m \leq k/2 - n/4 -c$, if we let $h$ be as above for $M=2^m$, then the following holds.
The function $\nm(x,y) = h(\dlog_g(x+y))$ is a $(k,\eps)$-non-malleable extractor for $\eps = O(n2^{m+n/4-k/2})$.
\end{theorem}

The main ingredient in our proof is Rao's generalization of Vazirani's XOR lemma
\cite{Rao:bourgain}.

\subsection{A generalized XOR lemma}

We now extend Rao's generalization of Vazirani's XOR lemma.  We need to modify his lemma because our output won't necessarily be uniform.

\begin{lemma}
\label{genxor}
For every positive integers $M \leq N$, the function
$h:\dbZ_N \to H=\dbZ_M$ defined above satisfies the following property.
Let $W,W'$ be any random variables on $\dbZ_N$ such that for all characters $\phi,\phi'$ on $\dbZ_N$ with $\phi$ nontrivial,
we have $|\expect_{(W,W')}[\phi(W)\phi'(W')]| \leq \alpha$.  Then $(h(W),h(W'))$ is $O(\alpha M \log N + M/N)$-close to the distribution
$(U,W')$, where $U$ is the uniform distribution on $H$ which is independent of~$W'$.
\end{lemma}

To prove Theorem~\ref{arbitraryM} assuming Lemma~\ref{genxor}, we set $N = q-1$, $(W,W') = (\dlog_g (X+Y), \dlog_g (X+\adv(Y)))$, and we condition on
$Y=y$.  
Note that for $\phi$ an additive character, the function $\chi(x) = \phi(\dlog_g (x))$ is a multiplicative character.  Therefore,
Theorem~\ref{charsum} shows that $((W,W')|Y=y)$ satisfies the hypotheses of Lemma~\ref{genxor} for some $\alpha_y$,
where $\expect_{y \leftarrow Y} [\alpha_y] \leq \alpha$ for $\alpha < q^{1/4} 2^{1-k/2} < 2^{n/4 + 2 - k/2}$.
Thus, by Lemma~\ref{genxor}, one finds that $((h(W),h(W'))|Y=y)$ is $O(\alpha_y M \log N + M/N)$-close to $((U,h(W'))|Y=y)$ for every $y$.
Since this expression is linear in $\alpha_y$, we conclude that $(h(W),h(W'),Y)$ is $O(\alpha M \log N + M/N)$-close to $(U,h(W'),Y)$, as required.

We now turn to the proof of Lemma~\ref{genxor}.
First note that Lemma~\ref{special-case} is a special case.
To handle $h$ in the case that $M\not | (q-1)$,
note that a character on a group $G$ has one Fourier coefficient $|G|$ and the rest 0.  We show that if the $\ell_1$-norm
of $\phi \circ h$ is not much bigger than this, then we get the desired conclusion.

\begin{lemma}
\label{lem:rao-plus}
Let $G$ and $H$ be finite abelian groups.  Let $(W,W')$ be a distribution on $G \times G$ with $|\expect_{(W,W')}[(\psi,\psi')(W,W')]| \leq \alpha$
for all nontrivial characters $\psi$ and all characters $\psi'$.  Let $h:G \to H$ be a function such that for every character $\phi$ of $H$, we have that
\[ \lone{\widehat{\phi \circ h}} \leq b|G|. \]
Then $\lone{(h(W),h(W')) - (h(U),h(W'))} \leq b\alpha |H|$.
\end{lemma}

\begin{proof}
Let $g:H \times H \to \dbC$ be the difference of distributions $(h(W),h(W'))-(h(U),h(W'))$,
and let $f:G \times G \to \dbC$ be the difference of distributions $(W,W')-(U,W')$.
By Lemma~\ref{special-case}, we have $|\linfty{f}| \leq \alpha$.
Let $\phi$ and $\phi'$ be any characters of $H$, with $\phi$ nontrivial.
Then
\begin{eqnarray*}
|\widehat{g}(\phi,\phi')| &=&
|\angles{(\phi,\phi'),g=(h(W),h(W'))-(h(U),h(W'))}|\\
&=& |\angles{(\phi,\phi') \circ h, f=(W,W') - (U,W')}|\\
&=& |\angles{\widehat{(\phi,\phi') \circ h}, \widehat{f}}|/|G|^2\\
&\leq& \lone{\widehat{(\phi,\phi') \circ h}} \linfty{\widehat{f}}/|G|^2\\
&\leq&  \lone{\widehat{(\phi \circ h,\phi' \circ h)}} \cdot \alpha/|G|^2.
\end{eqnarray*}
But now
\[ \lone{\widehat{(\phi \circ h,\phi' \circ h)}} = |\angles{\widehat{\phi \circ h},\widehat{\phi' \circ h}}| \leq \lone{\widehat{\phi \circ h}} \linfty{\widehat{\phi' \circ h}}
\leq (b |G|) |G|. \]
Putting these together yields $|\widehat{g}(\phi,\phi')| \leq b \alpha$.
When $\phi$ is trivial, as in Lemma~\ref{special-case}, one has $\widehat{g}(\phi,\phi') = 0$.
By (\ref{fourier-bound}), this implies $\lone{g} \leq |H| b \alpha$, as required.
\end{proof}

We bound $b$ using the following lemma by Rao, renormalized to our setting.

\begin{lemma}
Let $M<N$ be integers, and let $h:\dbZ_N \to \dbZ_M$ be the function $h(x) = x\mod M$.  Then for every character $\phi$ of $\dbZ_M$,
we have $\lone{\phi \circ h} = O(N\log N)$.
\end{lemma}

Thus, we may take $b=O(\log N)$ in Lemma~\ref{lem:rao-plus}.
Finally, we use the following simple lemma from Rao.

\begin{lemma}
Let $M<N$ be integers, and let $h:\dbZ_N \to \dbZ_M$ be the function $h(x) = x\mod M$.  Then for $U$ the uniform distribution on $\Z_N$,
we have that $h(U)$ is $2M/N$-close to the uniform distribution on $\Z_M$.
\end{lemma}

%% file: nmprivacy.bbl
\newcommand{\etalchar}[1]{$^{#1}$}
\begin{thebibliography}{DLWZ11}

\bibitem[BBR88]{bbr}
C.H. Bennett, G.~Brassard, and J.-M. Robert.
\newblock Privacy amplification by public discussion.
\newblock {\em SIAM Journal on Computing}, 17:210--229, 1988.

\bibitem[BKS{\etalchar{+}}05]{BarakKSSW05}
Boaz Barak, Guy Kindler, Ronen Shaltiel, Benny Sudakov, and Avi Wigderson.
\newblock Simulating independence: New constructions of condensers, {R}amsey
  graphs, dispersers, and extractors.
\newblock In {\em Proceedings of the 37th Annual ACM Symposium on Theory of
  Computing}, pages 1--10, 2005.

\bibitem[CG88]{cg:weak}
B.~Chor and O.~Goldreich.
\newblock Unbiased bits from sources of weak randomness and probabilistic
  communication complexity.
\newblock {\em SIAM Journal on Computing}, 17(2):230--261, 1988.

\bibitem[CKOR10]{ckor}
N.~Chandran, B.~Kanukurthi, R.~Ostrovsky, and L.~Reyzin.
\newblock Privacy amplification with asymptotically optimal entropy loss.
\newblock In {\em Proceedings of the 42nd Annual ACM Symposium on Theory of
  Computing}, pages 785--794, 2010.

\bibitem[CRS11]{CRS11}
Gil Cohen, Ran Raz, and Gil Segev.
\newblock Non-malleable extractors with short seeds and applications to privacy
  amplification.
\newblock {\em ECCC Report TR11-096}, 2011.

\bibitem[DKRS06]{dkrs}
Y.~Dodis, J.~Katz, L.~Reyzin, and A.~Smith.
\newblock Robust fuzzy extractors and authenticated key agreement from close
  secrets.
\newblock In {\em CRYPTO}, pages 232--250, 2006.

\bibitem[DKSS09]{DvirKSS09}
Zeev Dvir, Swastik Kopparty, Shubhangi Saraf, and Madhu Sudan.
\newblock Extensions to the method of multiplicities, with applications to
  kakeya sets and mergers.
\newblock In {\em Proceedings of the 50th Annual IEEE Symposium on Foundations
  of Computer Science}, 2009.

\bibitem[DLWZ11]{dlwz-focs}
Y.~Dodis, X.~Li, T.D. Wooley, and D.~Zuckerman.
\newblock Privacy amplification and non-malleable extractors via character
  sums.
\newblock In {\em Proceedings of the 52nd Annual IEEE Symposium on Foundations
  of Computer Science}, 2011.

\bibitem[DO03]{DodO}
Y.~Dodis and R.~Oliveira.
\newblock On extracting private randomness over a public channel.
\newblock In {\em RANDOM 2003, 7th International Workshop on Randomization and
  Approximation Techniques in Computer Science}, pages 252--263, 2003.

\bibitem[DORS08]{dors}
Y.~Dodis, R.~Ostrovsky, L.~Reyzin, and A.~Smith.
\newblock Fuzzy extractors: How to generate strong keys from biometrics and
  other noisy data.
\newblock {\em SIAM Journal on Computing}, 38:97--139, 2008.

\bibitem[DW09]{DodW}
Y.~Dodis and D.~Wichs.
\newblock Non-malleable extractors and symmetric key cryptography from weak
  secrets.
\newblock In {\em Proceedings of the 41st Annual ACM Symposium on Theory of
  Computing}, 2009.

\bibitem[GUV09]{GuruswamiUV09}
Venkatesan Guruswami, Christopher Umans, and Salil Vadhan.
\newblock Unbalanced expanders and randomness extractors from
  {P}arvaresh-{V}ardy codes.
\newblock {\em Journal of the ACM}, 56(4), 2009.

\bibitem[HB78]{hb78}
D.R. Heath-Brown.
\newblock Almost-primes in arithmetic progressions and short intervals.
\newblock {\em Math. Proc. Cambridge Philos. Soc.}, 83:357--375, 1978.

\bibitem[HB92]{hb92}
D.R. Heath-Brown.
\newblock Zero-free regions for {Dirichlet} {$L$}-functions, and the least
  prime in an arithmetic progression.
\newblock {\em Proc. London Math. Soc.}, 64:265--338, 1992.

\bibitem[KR08]{kr:robust-fuzzy}
B.~Kanukurthi and L.~Reyzin.
\newblock An improved robust fuzzy extractor.
\newblock In {\em SCN}, pages 156--171, 2008.

\bibitem[KR09]{kr:agree-close}
B.~Kanukurthi and L.~Reyzin.
\newblock Key agreement from close secrets over unsecured channels.
\newblock In {\em EUROCRYPT}, pages 206--223, 2009.

\bibitem[MW97]{MW97}
Ueli~M. Maurer and Stefan Wolf.
\newblock Privacy amplification secure against active adversaries.
\newblock In {\em CRYPTO '97}, 1997.

\bibitem[NZ96]{nz}
N.~Nisan and D.~Zuckerman.
\newblock Randomness is linear in space.
\newblock {\em Journal of Computer and System Sciences}, 52(1):43--52, 1996.

\bibitem[Rao07]{Rao:bourgain}
A.~Rao.
\newblock An exposition of {Bourgain's} 2-source extractor.
\newblock Technical Report TR07-034, Electronic Colloquium on Computational
  Complexity, 2007.

\bibitem[Raz05]{Raz05}
Ran Raz.
\newblock Extractors with weak random seeds.
\newblock In {\em Proceedings of the 37th Annual ACM Symposium on Theory of
  Computing}, pages 11--20, 2005.

\bibitem[RW03]{RW03}
R.~Renner and S.~Wolf.
\newblock Unconditional authenticity and privacy from an arbitrarily weak
  secret.
\newblock In {\em CRYPTO}, pages 78--95, 2003.

\bibitem[Sch76]{Sch1976}
W.M. Schmidt.
\newblock {\em Equations over Finite Fields. An Elementary Approach}, volume
  536 of {\em Lecture Notes in Mathematics}.
\newblock Springer-Verlag, 1976.

\bibitem[Ter99]{Ter}
A.~Terras.
\newblock {\em Fourier Analysis on Finite Groups and Applications}.
\newblock Cambridge University Press, 1999.

\bibitem[Wei48]{Wei1948}
A.~Weil.
\newblock On some exponential sums.
\newblock {\em Proceedings of the National Academy of Sciences}, 34:204--207,
  1948.

\bibitem[Xyl11]{Xyl}
T.~Xylouris.
\newblock On the least prime in an arithmetic progression and estimates for the
  zeros of {Dirichlet $L$-functions}.
\newblock {\em Acta Arithmetica}, 150:65--91, 2011.

\bibitem[Zuc07]{Zuc07}
David Zuckerman.
\newblock Linear degree extractors and the inapproximability of {Max Clique}
  and {Chromatic Number}.
\newblock In {\em Theory of Computing}, pages 103--128, 2007.

\end{thebibliography}
